\documentclass[a4paper, 10pt, reqno]{amsart}
\usepackage{setspace}
\usepackage{amsthm}
\usepackage{setspace}
\usepackage{parskip}
\usepackage{enumitem}
\usepackage{hyperref}
\usepackage{scalerel,stackengine}
\stackMath
\newcommand\reallywidehat[1]{%
\savestack{\tmpbox}{\stretchto{%
  \scaleto{%
    \scalerel*[\widthof{\ensuremath{#1}}]{\kern-.6pt\bigwedge\kern-.6pt}%
    {\rule[-\textheight/2]{1ex}{\textheight}}
  }{\textheight}%
}{0.5ex}}%
\stackon[1pt]{#1}{\tmpbox}%
}
\usepackage[margin=1in]{geometry}
\usepackage{amssymb,amsmath,amsfonts}
\usepackage{graphicx,color}
\usepackage{epstopdf}
\usepackage{enumerate}
\usepackage{caption}
\usepackage{subfig}
\usepackage{multirow}
\usepackage{mathrsfs}
\usepackage{tikz}
\usepackage{booktabs}
\usepackage{colortbl}
\usepackage{amsmath}
\usepackage{amsthm}
\usepackage{amssymb}
\usepackage{graphicx}
\usepackage{caption}
\usepackage{subcaption}
\usepackage{relsize}
\usepackage{xcolor}
\usepackage{mismath}
\makeatletter \oddsidemargin  -.1in \evensidemargin -.1in
\usepackage[active]{srcltx}
\usepackage{comment}
\usepackage[mathscr]{euscript}
\usepackage{placeins}
\usepackage[symbol]{footmisc}
\usetikzlibrary{shapes.geometric}

\definecolor{Green}{rgb}{0.0, 0.5, 0.}
\definecolor{Red}{rgb}{0.7, 0.2, 0}
\newcommand\del[1]{}

\numberwithin{equation}{section}

\newtheorem{theorem}{Theorem}[section]
\newtheorem{definition}[theorem]{Definition}
\newtheorem{corollary}[theorem]{Corollary}
\newtheorem{proposition}[theorem]{Proposition}
\newtheorem{lemma}[theorem]{Lemma}

\newtheorem{remark}{Remark}

\numberwithin{figure}{section}

\begin{document}
\author{\makebox[\textwidth]
{Brijesh Kumar Jha$^\dag$, Subhra Sankar Dhar$^\ddagger$, \and Akash Ashirbad Panda$^\star$}
}

\thanks{\\$^\dag$Discipline of Data Science, School of Interwoven Arts and Sciences, Krea University, Campus: 5655, Central Expressway, Sri City-517646, Andhra Pradesh  ({\tt brijeshkumar.jha@krea.edu.in}).
	\\$^\ddagger$ Department of Mathematics and Statistics, IIT Kanpur, Kanpur-208016, India ({\tt subhra@iitk.ac.in}).
	\\$^\star$Department of Mathematics, School of Basic Sciences, Indian Institute of Technology Bhubaneswar, Khorda-752050, Odisha ({\tt akashpanda@iitbbs.ac.in}).	
}

\title{{Least Square Estimation: SDE Perturbed by L\'evy Noise with Sparse Sample Paths}}

\begin{abstract}
This article investigates the least squares estimators (LSE) for the unknown parameters in stochastic differential equations (SDEs) that are affected by L\'evy noise, particularly when the sample paths are sparse. Specifically, given $n$ sparsely observed curves related to this model, we derive the least squares estimators for the unknown parameters: the drift coefficient, the diffusion coefficient, and the jump-diffusion coefficient. We also establish the asymptotic rate of convergence for the proposed LSE estimators. Additionally, in the supplementary materials, the proposed methodology is applied to a benchmark dataset of functional data/curves, and a small simulation study is conducted to illustrate the findings.


\end{abstract}

\maketitle


\medskip

 	
\pagestyle{myheadings}
\thispagestyle{plain}
\markboth{B. K. Jha, S. S. Dhar, and A. A. Panda}{Least Square Estimation: SDEs Perturbed by L\'evy Noise with Sparse Sample Paths} 
 	
\medskip

\section{Introduction}\label{sec-1}

Stochastic differential equations (SDEs) provide a powerful framework for modeling random dynamics in fields such as mathematical finance, population biology, and statistical physics (see, e.g., \cite{Oksendal_2003, Mao_2011, Jones_2017}). An SDE incorporates stochastic forcing, making its solution a stochastic process. From an analytical perspective, driving noise is often modeled as white noise, formally the time derivative of a Brownian motion; in a probabilistic setting, SDEs are frequently used to describe diffusion processes: continuous-time Markov processes with almost surely continuous sample paths.

In many applications, however, the driving noise exhibits discontinuities (jumps), which cannot be captured by Brownian motion alone. L\'evy processes (see, e.g., \cite{Applebaum_2004, Sato_2010}) form the canonical class of such jump processes and are fundamental building blocks for continuous-time models with discontinuous trajectories. This class includes, as special cases, the Wiener process, the Poisson process, (Moran-) Gamma process, the Pascal process, and the Meixner process (see, e.g., \cite{Bruss_1991, Klenke_2008, Mazzola_2011}). Their flexibility has led to widespread applications in finance, insurance, physics, and engineering, some of which are outlined in Section~\ref{SM}.

Despite their importance, the statistical estimation of SDEs driven by L\'evy noise—particularly under minimal regularity conditions and in nonparametric settings—remains considerably less developed than in the purely Brownian framework. Addressing this gap, the present work develops and analyzes estimators for drift and diffusion terms in SDEs with additive L\'evy perturbations, providing rigorous consistency and convergence results supported by simulations and real-data applications.

In this work, we focus on SDEs driven by additive L\'evy processes, which can be expressed in the following general form.

\vspace{0.1cm}

\subsection{The stochastic model}\label{SM}

Let $\big(X(t)\big)_{0\le t\le T}$ be a real-valued stochastic process defined on a
filtered probability space $\mathfrak P := (\Omega,\mathcal F,\mathbb F,\mathbb P),$ where $\mathbb F=\{\mathcal F_t\}_{0\le t\le T}$ denotes the filtration generated by
$X$, that is, $\mathcal F_t := \sigma\big(X(s):0\le s\le t\big),$ augmented to satisfy the usual conditions of completeness and right-continuity
(see \cite[Section~2.1.1]{Applebaum_2004}).

In this paper, we consider the following stochastic differential equation driven
by Lévy noise with small jumps:
\begin{align}\label{SDE-1}
dX(t)
= \mu(t)\,X(t)\,dt
+ \sigma(t)\,dW(t)
+ \int_{|y|\le1}\xi(t)\,\tilde\eta(dt,dy), \qquad t\in[0,T].
\end{align}
Here,
\begin{itemize}
\item $\mu(\cdot)$, $\sigma(\cdot)$, and $\xi(\cdot)$ are real-valued deterministic
functions defined on $[0,T]$, representing the drift, diffusion, and jump
amplitude coefficients, respectively;
\item $\{W(t)\}_{t\ge0}$ is a one-dimensional standard Brownian motion;
\item $\tilde\eta(dt,dy)=\eta(dt,dy)-\nu(dy)\,dt$ is the compensated Poisson random
measure associated with a Poisson random measure $\eta$ on
$[0,\infty)\times\mathbb R$ with Lévy measure $\nu$ satisfying
\[
\int_{\mathbb R}(1\wedge |y|^2)\,\nu(dy)<\infty;
\]
\item the processes $W$ and $\tilde\eta$ are independent and independent of
$\mathcal F_0$.
\end{itemize}

Equation \eqref{SDE-1} describes a linear stochastic system with additive Lévy
noise and time-dependent coefficients. Throughout the paper, we restrict
attention to small jumps, which is sufficient for the development of the moment
equations and estimation procedures considered here.

\subsection*{A remark on the jump coefficient $\xi$}

Throughout this work, we restrict attention to jump amplitudes of the form $\xi=\xi(t),$ that is, independent of the jump size $y$. This assumption is adopted primarily
for notational simplicity and to highlight the core ideas of our estimation
framework, rather than due to any intrinsic mathematical limitation.

Indeed, all moment identities and estimation procedures developed in the later
Sections remain valid when the
jump coefficient depends on both time and jump size, $\xi=\xi(t,y)$, provided
that $\xi(\cdot,\cdot)$ is predictable and satisfies the standard
square-integrability condition
\[
\int_0^1\int_{|y|\le1}\mathbb E\big[\xi^2(t,y)\big]\,\nu(dy)\,dt<\infty.
\]
In this more general setting, the quantity $\nu(K)\,\xi^2(t)$ appearing in the
second-moment equation is naturally replaced by
\[
\int_{|y|\le1}\xi^2(t,y)\,\nu(dy),
\]
and all subsequent derivations proceed verbatim with this substitution.

Consequently, the restriction $\xi=\xi(t)$ does not reduce the generality of the
theoretical results, but rather yields a cleaner presentation and more explicit
estimators. Extensions to state- or jump-size–dependent jump amplitudes can be
handled in a completely analogous manner.



The above system \eqref{SDE-1} can be interpreted as an SDE driven by continuous noise interspersed with jumps, modeled via Lévy processes. Such a framework naturally arises in various applications where both small continuous fluctuations and sudden large shocks are present. For example, Merton \cite{Merton_1976} introduced a jump-diffusion model for option pricing in which underlying asset returns are generated by a mixture of continuous Brownian motion and jump processes. This model, now widely known as the Merton model, extends the classical Black–Scholes framework while preserving many of its desirable analytical features. Kou \cite{Kou_2002} proposed the double exponential jump-diffusion model, which yields closed-form solutions for a broad class of option pricing problems, including standard European options, path-dependent options, and interest rate derivatives, while offering a realistic balance between empirical accuracy and mathematical tractability. Beyond financial mathematics, Lévy-driven models have also been applied to complex fluid dynamics. For example, Brze\'zniak et al.\ \cite{ZB+UM+AAP_2019} studied a stochastic evolution equation for nematic liquid crystals perturbed by Lévy noise, while Manna and Panda \cite{Manna+Panda_2021} analyzed the constrained Navier–Stokes equations under L\'evy perturbations. In both cases, the authors emphasized the physical relevance of incorporating jump noise to more accurately represent abrupt changes in the underlying physical systems.

\vspace{0.2cm}

Let us now turn to the problems of interest associated with \eqref{SDE-1}. 
It is well established in the literature on jump processes that the solution of \eqref{SDE-1} is determined by the functions $\mu(t)$, $\sigma(t)$, and $\xi(t)$ (see, e.g., \cite[Section~6.2]{Applebaum_2004}). Consequently, the properties of the L\'evy process governed by \eqref{SDE-1} can be fully characterized through $\mu(t)$, $\sigma(t)$, and $\xi(t)$ (see, e.g., \cite[Section~3.2]{Kunita_2019}). 
Therefore, estimating these quantities is of clear importance and analyzing the statistical properties of their estimators is of significant research interest. 
However, it should be noted that the estimation of $\mu(t)$, $\sigma(t)$, and $\xi(t)$ critically depends on the observation scheme of the diffusion process. 
For example, the data may consist of multiple replications of processes satisfying \eqref{SDE-1}, where each sample path (a continuous-time trajectory) is observed only at discrete time points, possibly contaminated with measurement error.

\vspace{0.2cm}

In our data setting, the number of replications $n$ can be large, while the number of time observations $r$ may or may not depend on $n$, and may also grow with $n$. From such samples, we estimate $\mu(t)$, $\sigma(t)$, and $\xi(t)$ via their relationships in a system of differential equations that involve additional moment terms of the process (see Proposition~\ref{prop-m,sigma-corrected}), and we establish convergence results for the proposed estimators. 

Previously, Mohammadi et al. \cite{Mohammadi2024} studied least-squares estimators for drift and diffusion in SDEs with only Brownian noise, without jumps. In contrast, we incorporate jumps, modeled by a compensated Poisson random measure, which significantly complicates the analysis (see Section~\ref{ND}). Consequently, the proofs of our main results (Theorems~\ref{Thm-Main-1} and \ref{Thm-main-2}) differ completely from those of \cite{Mohammadi2024,B3}. The key contribution of this work is therefore the introduction of jumps in the model, together with the development of estimators and asymptotic theory for $\mu(t)$, $\sigma(t)$, and $\xi(t)$ in this more general setting.

\vspace{0.2cm}

\subsection{Novelties and technical challenges}\label{ND}

The stochastic differential equation \eqref{SDE-1} belongs to a class of linear,
time-inhomogeneous systems driven by both Gaussian and L\'evy-type jump noise.
The primary objective of this work is to estimate the time-varying drift
$\mu(t)$, diffusion $\sigma(t)$, and jump-diffusion coefficient $\xi(t)$ from
noisy and sparsely observed data.

More precisely, we consider $n$ independent sample paths
$\{X_i(t)\}_{i=1}^n$ solving \eqref{SDE-1}, where each trajectory is observed at
$r\ge2$ random time points with additive measurement error,
\[
Y_{ij} = X_i(T_{ij}) + U_{ij}, \qquad i=1,\ldots,n,\; j=1,\ldots,r,
\]
with $T_{ij}\in[0,T]$ denoting random observation times and $\{U_{ij}\}$ being
independent noise variables; see \eqref{eq3} for precise assumptions.

The least squares framework is particularly natural in this setting, as it is
well aligned with quadratic loss minimization and remains computationally
tractable under irregular and noisy sampling schemes. Our methodology draws
inspiration from functional data analysis (FDA), where information is pooled
across multiple trajectories. However, classical FDA techniques rely heavily on
smooth sample paths and are therefore not directly applicable to data generated
by L\'evy-driven SDEs, whose trajectories are typically c\`adl\`ag and may exhibit
frequent jumps. The methods developed here explicitly account for this lack of
pathwise smoothness and remain valid in the presence of discontinuities.

\medskip
\noindent
Our approach proceeds in three main steps:
\begin{itemize}
  \item \emph{Moment-based deterministic characterization.}
  We first derive systems of deterministic ODEs and PDEs that link the unknown
  coefficients $\mu(t)$, $\sigma(t)$, and $\xi(t)$ to the first- and second-order
  moment functions of $X(t)$ and their derivatives. Since \eqref{SDE-1} is driven
  by L\'evy noise, these systems are obtained using the It\^o--L\'evy formula
  (see Section~\ref{Sec-Ito-lemma}). The resulting equations, presented in
  Section~\ref{sec-PDEs}, encode the local dynamics of the process through global
  moment information and form the analytical backbone of our estimation strategy.

  \item \emph{Nonparametric estimation of moments.}
  We next construct nonparametric estimators for the mean and covariance
  functions, together with their derivatives, using local polynomial regression.
  In particular, the covariance surface and its partial derivatives are estimated
  by pooling empirical correlations across sample paths observed at random and
  possibly sparse time points, following ideas introduced by Yao et al.~\cite{B6}.
  Key modifications are required to handle the roughness induced by jumps and to
  address bias effects near the diagonal; see
  Section~\ref{sec-estimators-definition}.

  \item \emph{Coefficient recovery via plug-in procedures.}
  Finally, the estimated moment functions are substituted into the deterministic
  ODE/PDE systems derived earlier, yielding explicit plug-in estimators for
  $\mu(t)$, $\sigma(t)$, and $\xi(t)$. This step bridges the gap between global,
  pooled statistical information and local coefficient identification; see
  Section~\ref{subsec-estimators}.
\end{itemize}

\medskip
\noindent
We establish strong consistency and almost sure uniform convergence rates for the
proposed estimators as the number of replications $n\to\infty$. Importantly, the
number of observations per trajectory, $r$, is allowed to remain fixed or to vary
with $n$, covering both sparse and dense functional data regimes. To the best of
our knowledge, this work represents one of the first systematic studies of
nonparametric coefficient estimation for L\'evy-driven SDEs under noisy and
irregular sampling. The proposed framework does not rely on stationarity
assumptions, imposes no parametric restrictions on the coefficients, and allows
for simultaneous estimation of $\mu(t)$, $\sigma(t)$, and $\xi(t)$. Moreover, the
resulting procedures are computationally simple and scalable.

\medskip
\noindent
\textbf{Technical challenges.}
Several nontrivial difficulties arise in this setting:
\begin{itemize}
  \item
  The presence of L\'evy jumps fundamentally alters the analytical structure of
  the problem. Unlike Brownian-driven systems, where sample paths are continuous,
  jump noise introduces discontinuities and nonlocal effects, rendering many
  classical tools inapplicable. This necessitates working within the
  It\^o--L\'evy framework and carefully handling compensated Poisson integrals,
  which exhibit singular and nonlocal behavior.

  \item
  Due to jump perturbations, the process $X(t)$ is only c\`adl\`ag, and its second
  moments must be characterized through two-time covariance functions rather than
  pointwise derivatives. This complicates both the derivation of deterministic
  moment equations and the construction of derivative estimators near the
  diagonal.

  \item
  Measurement error contaminates diagonal covariance terms, requiring special
  treatment to avoid bias and to ensure consistent estimation of partial
  derivatives of the covariance surface. This issue is particularly acute in the
  presence of jumps, where local variability is amplified.

  \item
  The identification of the jump-diffusion coefficient $\xi(t)$ relies on subtle
  interactions between temporal derivatives of the covariance function and
  integral representations derived from the moment equations. Ensuring
  identifiability and stability of these estimators requires careful analytical
  control.
\end{itemize}

\subsection{Numerical Study and Real Applications:} \label{RA} 
As it is indicated in the beginning of Section \ref{sec-1}, \eqref{SDE-1} can be fitted to many real life data, and using least squares estimators of the unknown parameters $\mu(.)$, $\sigma(.)$ and $\xi(.)$, one can capture the central tendency of the data. In fact, strictly speaking, many functional data (see, e.g., \cite{Ramsey2005}) with jump(s) in some sense can be analysed by least squares estimators of the unknown parameters involved in \eqref{SDE-1}. In this work, well-known Canadian weather data (see, e.g., \cite{Ramsey2005}) is studied, and in the context of various issues related to functional/infinite dimensional data, this data set was earlier analysed by \cite{Dette2024}. The analysis carried out on this data set gives insight about many features about the temperature and the amount of precipitation in those locations in Canada.   
Moreover, we also carry out some simulation studies to check the performance of the proposed estimators of $\mu(t)$, $\sigma(t)$ and $\xi(t)$. 

\vspace{0.2cm}

\subsection{Organization of the Article} 

Section 1 provides motivation for studying SDEs with jump noise, highlighting the key challenges and novelties of the problem. In Section 2, we present the necessary assumptions and preliminaries, including definitions related to c\`adl\`ag and adapted processes, along with the It\^o-L\'evy formula, which underpin our analytical framework. Section 3 is devoted to the derivation of auxiliary deterministic systems of ODEs/PDEs that are used as part of our methodology to estimate the drift and diffusion components. In Section 4, we introduce the proposed non-parametric estimators for these coefficients. The main theoretical results, including the consistency and convergence rates of the estimators, are established in Section 5. Finally, Section 6 provides a detailed proof of Theorem 5.2. In the supplementary materials, we have presented simulation studies and real-data analysis. 

\vspace{0.3cm}

\section{Mathematical Framework and Preliminaries}

We recall the following definitions.
\begin{definition}\label{cadlag}
    Let $I=[a, b]$ be an interval in $\mathbb{R}^{+}$. A mapping $f: I \rightarrow \mathbb{R}^d$ is said to be c\`adl\`ag if, for all $t \in(a, b], f$ has a left limit at $t$ and $f$ is right-continuous at $t$, i.e.,
    \begin{itemize}
        \item for all sequences $\left(t_n\right)_{n \in \mathbb{N}}$ in $I$ with $t_n<t$ and $\displaystyle\lim _{n \rightarrow \infty} t_n=t$, we have that $\displaystyle\lim _{n \rightarrow \infty} f\left(t_n\right)$ exists;
        \item for all sequences $\left(t_n\right)_{n \in \mathbb{N}}$ in $I$ with $t_n \geq t$ and $\displaystyle\lim _{n \rightarrow \infty} t_n=t$, we have that $\displaystyle\lim _{n \rightarrow \infty} f\left(t_n\right)=f(t)$.
    \end{itemize}
\end{definition}

\begin{definition}\label{adapted}
A stochastic process is said to be adapted to a filtration $\mathbb{F}$, if $X(t)$ is a random variable on $\mathcal{F}_t,$  $\forall \,t \geq 0$, i.e., if $X(t)$ is $\mathcal{F}_t$-measurable.
\end{definition}

\medskip
\noindent
{\bf Assumptions.}
Throughout this paper, we work on the finite time interval $[0,1]$ and impose the following assumptions.

\medskip
\noindent
{\bf (A1)}\,
(Regularity of coefficients.)
The drift, diffusion, and jump diffusion coefficients
\[
\mu(\cdot),\ \sigma(\cdot),\ \xi(\cdot)
\]
are deterministic functions belonging to $C^d([0,1];\mathbb{R})$ for some integer $d\ge1$.

\medskip
\noindent
{\bf (A2)}\,
(Square-integrability.)
The coefficients satisfy the integrability conditions
\[
\int_0^1 \sigma^2(t)\,dt < \infty,
\qquad
\int_0^1 \xi^2(t)\,dt < \infty.
\]

\medskip
\noindent
{\bf (A3)}\,
(Finite activity of small jumps.)
The Lévy measure $\nu$ satisfies
\[
\nu(K)<\infty,
\qquad K:=\{y\in\mathbb R:\ |y|\le1\}.
\]
This ensures that the compensated Poisson integral with jump amplitude $\xi(t)$ is well defined and has finite second moments.




As the well-posedness of the problem, we rewrite the SDE \eqref{SDE-1} perturbed by L{\'e}vy noise in the integral form as
\begin{align}\label{SDE-integral}
X(t)= X(0)+ \int_0^1 \mu(t) X(t)\,dt + \int_0^1 \sigma(t)\, dW(t)+ \int_0^1 \int_{|y|\leq 1} \xi(t)\,\tilde{\eta}(dt,dy)\,. 
\end{align}
We have the following well-posedness result for the system \eqref{SDE-integral}.

\begin{theorem}
     Under the assumption {\bf (A1)-(A3)},  there exists a unique solution $X = (X(t),t \geq 0)$ to the system \eqref{SDE-integral} with the standard initial condition, i.e., $\mathbb{E}\big[ |X_0|^2\big] < \infty$. Moreover, the process $X(.)$ is c\`adl\`ag and adapted as defined in the Definitions \ref{cadlag} and \ref{adapted}, respectively.
\end{theorem}
\begin{proof}
    The proof of existence and uniqueness is based on the technique of Picard's iteration; see Applebaum \cite[Chapter 6]{Applebaum_2004}.
\end{proof}

\vspace{0.2cm}

In the next subsection, we will state the celebrated It\^o Lemma/Formula, which is an essential tool in stochastic calculus, analogous to the chain rule of ordinary calculus. It helps in determining the differential of a function of a stochastic process, effectively accounting for the inherent randomness of such processes. In this work, as part of the methodology, we aim to derive the systems of ODEs/PDEs that relate $\mu(t)$, $\sigma(t)$ and $\xi(t)$ of the process $X(t)$ to the various moment functions of $X(t)$, along with their first derivatives. In order to derive such systems, an appropriate version of the Ito formula must be applied. 

\subsection{The It\^o Lemma/Formula \cite{Applebaum_2004}}\label{Sec-Ito-lemma}

We begin by recalling the It\^o formula for L\'evy-type stochastic integrals.  
Let $\mathrm{X}$ be a $d$-dimensional process satisfying the stochastic differential equation
\begin{align}\label{general-Levy-integral}
d \mathrm{X}(t) 
&= \mu(t)\, dt + \sigma(t)\, d W(t)
   + \int_{|x|<1} \xi(t, x)\, \tilde{\eta}(dt, dx) 
   + \int_{|x|\geq 1} \Xi(t, x)\, \eta(dt, dx),
\end{align}
where $\mu=(\mu_j^i)$ is a $d\times m$ matrix, $W=(W^1,\dots,W^m)$ is an $m$-dimensional standard Brownian motion, and $\eta$ is a Poisson random measure on $\mathbb{R}^+ \times (\mathbb{R}^d\setminus\{0\})$ with compensator $\tilde{\eta}$ and intensity measure $\nu$, assumed to be a L\'evy measure.  

Let $E=\{x\in \mathbb{R}^d : |x|<1\}\setminus \{0\}$.  
We define $\mathcal{P}_2(T,E)$ as the collection of equivalence classes of mappings  
\[
J:[0,T]\times E\times \Omega \to \mathbb{R}
\]
such that:  
\begin{itemize}
    \item $J$ is predictable, i.e., for each fixed $x\in E$, the map $(t,\omega)\mapsto J(t,x,\omega)$ is $(\mathcal{F}_t)$-predictable;  
    \item $\mathbb{P}\!\left(\int_0^T \int_E |J(t,x)|^2 \,\rho(dt,dx) < \infty \right) = 1$.  
\end{itemize}
For the case $E=\{0\}$, we write $\mathcal{P}_2(T,\{0\}):=\mathcal{P}_2(T)$, where $\mathcal{P}_2(T)$ denotes the set of all predictable mappings $J:[0,T]\times \Omega\to \mathbb{R}$ satisfying  
\[
\mathbb{P}\!\left(\int_0^T |J(t)|^2\, dt < \infty \right) = 1.
\]
We assume that for each $1\leq i \leq d$, $1\leq j\leq m$, and $t\geq 0$, the coefficients satisfy
\[
|\sigma^i|^{1/2},\ \mu_j^i \in \mathcal{P}_2(T), 
\qquad \xi^i \in \mathcal{P}_2(T,E),
\]
and $\Xi$ is predictable.  

\noindent
It is convenient to decompose $\mathrm{X}$ into its continuous and discontinuous components:
\[
d\mathrm{X}_{\mathrm{c}}(t) = \mu(t)\,dt + \sigma(t)\, dW(t),
\]
\[
d\mathrm{X}_{\mathrm{d}}(t) = \int_{|x|<1} \xi(t,x)\, \tilde{\eta}(dt,dx) 
+ \int_{|x|\geq 1} \Xi(t,x)\, \eta(dt,dx).
\]
Thus, for each $t\geq 0$, we may write
\[
\mathrm{X}(t) = \mathrm{X}(0) + \mathrm{X}_{\mathrm{c}}(t) + \mathrm{X}_{\mathrm{d}}(t).
\]

\medskip
This decomposition separates the evolution of $\mathrm{X}$ into a continuous martingale part driven by Brownian motion and a purely discontinuous jump part governed by the Poisson random measure, a perspective that will be essential for applying the It\^o formula in the sequel.

Let us denote for each $1 \leq i \leq j$, the quadratic variation process as $\left(\left[\mathrm{X}_c^i, \mathrm{X}_c^j\right](t), t \geq 0\right)$, by
$$
\left[\mathrm{X}_c^i, \mathrm{X}_c^j\right](t)=\sum_{k=1}^m \int_0^t \mu_k^i(s) \mu_k^j(s) d s\,.
$$
For the following result, we refer to \cite[Theorem 4.4.7]{Applebaum_2004}.

\begin{proposition}[It\^o Lemma for L\'evy-type stochastic integral]\label{Ito-lemma}
    If $\mathrm{X}$ is a Lévy-type stochastic integral of the form \eqref{general-Levy-integral}, then, for each $f \in C^2\left(\mathbb{R}^d\right), t \geq 0$, the following holds
$$
\begin{aligned}
f(\mathrm{X}(t))-f(\mathrm{X}(0)) &= \int_0^t \partial_i f(\mathrm{X}(s-)) d \mathrm{X}_c^i(s)+\frac{1}{2} \int_0^t \partial_i \partial_j f(\mathrm{X}(s-)) d\left[\mathrm{X}_{\mathrm{c}}^i, \mathrm{X}_{\mathrm{c}}^j\right](s) \\
&\quad+\int_0^t \int_{|x| \geq 1} \Big[f(\mathrm{X}(s-)+ \Xi(s, x))-f(\mathrm{X}(s-)) \Big] \eta(d s, d x) \\
&\quad+\int_0^t \int_{|x|<1} \Big[ f(\mathrm{X}(s-) + \xi(s, x)) - f(\mathrm{X}(s-)) \Big] \tilde{\eta}(d s, d x) \\
&\quad+\int_0^t \int_{|x|<1} \Big[ f(\mathrm{X}(s-)+ \xi(s, x))- f(\mathrm{X}(s-))) - \xi^i(s, x) \partial_i f(\mathrm{X}(s-)) \Big] \nu(d x) d s\,.
\end{aligned}
$$
with probability 1.
\end{proposition}

\subsection{Measurement Scheme}\label{subsec-measurement}

Assume {\bf (A1)--(A3)}.  
Let $\{X_i(t),\, t\in[0,1]\}_{1\leq i\leq n}$ be i.i.d. processes satisfying
\begin{equation}\label{2.3}
    dX_i(t) = \mu(t) X_i(t)\,dt + \sigma(t)\,dW_i(t) 
              + \int_{|y|\leq 1} \xi(t)\,\tilde{\eta}_i(dt,dy), 
    \qquad i=1,\dots,n,
\end{equation}
where $\{W_i\}$ are independent Brownian motions,  
$\{\tilde{\eta}_i\}$ independent Poisson random measures, and  
$\{X_i(0)\}$ independent initial conditions, all mutually independent.  
\noindent
We observe
\begin{equation}\label{eq3}
    Y_i(T_{ij}) = X_i(T_{ij}) + U_{ij}, 
    \qquad i=1,\dots,n,\ j=1,\dots,r(n),
\end{equation}
under the following conditions:
\begin{itemize}
    \item[(a)] $\{U_{ij}\}$ are i.i.d. centered errors with variance $\varrho^2<\infty$;
    \item[(b)] $\{T_{ij}\}$ form a triangular array of design points, strictly increasing in $j$, independent of $\{X_i,U_{ij}\}$, and sampled from a density bounded away from zero on $[0,1]$;
    \item[(c)] $r(n)\geq 2$, with no further restriction on its growth.
\end{itemize}

Thus the data consist of noisy observations $\{Y_{ij}\}$ of the latent processes $\{X_{ij}\}$ at random design points,  
where $r(n)$ determines the sampling frequency, ranging from sparse to dense designs.

Furthermore, we assume that the domain $[0, 1]$ is sampled relatively evenly, i.e.,

\noindent
{\bf (A4)}\, There exists $\widetilde{C} >0$ such that 
$$
0 < \mathbb{P}\left(T_{i j} \in [a, b]\right) \leq \widetilde{C}(b-a), \quad \forall i, j \ \mbox{ and } \ 0 \leq a < b \leq 1\,.
$$

We also make the  assumptions about the initial distribution of the diffusion process $X(t)$ and the $\alpha$-th moment of the centered measurement errors:

\noindent
{\bf (A5)}\, For some $\alpha >0$ (which will be specified later), 
$$
\mathbb{E} \big[ |X(0)|^\alpha \big] < \infty \quad \mbox{ and } \quad \mathbb{E} \big[ \left| U_{i j}\right|^\alpha \big] < \infty\,.
$$

\vspace{0.2cm}

\subsection{Methodology}

The local dynamics of the stochastic process in \eqref{SDE-1} are governed by the
time-varying drift, diffusion, and jump-diffusion coefficients
$\mu(\cdot)$, $\sigma(\cdot)$, and $\xi(\cdot)$. In classical settings,
these coefficients admit an infinitesimal interpretation in terms of
conditional moments. For instance, in time-homogeneous diffusion models without
jumps, one formally has
\[
\mu(y) = \lim_{h \to 0} \frac{1}{h}\,\mathbb{E}\!\left[X_{t+h}-y \mid X_t=y\right], 
\qquad 
\sigma^2(y) = \lim_{h \to 0} \frac{1}{h}\,\mathbb{E}\!\left[(X_{t+h}-y)^2 \mid X_t=y\right].
\]
Such representations motivate a large class of inference procedures based on
estimating local increments of the process, which typically requires dense or
high-frequency observations of individual sample paths.

However, in the present setting this local approach is no longer viable.
First, the data are assumed to be sparsely and irregularly observed in time, so
fine-scale increments cannot be directly recovered. Second, the presence of
L\'evy jump noise fundamentally alters the local behavior of the process:
increments may contain large discontinuities, and the classical infinitesimal
variance decomposition into drift and diffusion components no longer applies
without additional structure. As a result, traditional local estimation
strategies break down.

\medskip
\noindent
To address these challenges, we adopt a global, moment-based methodology that
indirectly recovers local coefficients from pooled information across sample
paths. The central idea is to exploit deterministic relationships between the
coefficients of the SDE and the first- and second-order moment structure of the
process, which remain well-defined even under sparse sampling and jump
perturbations.

Our methodology proceeds in three main steps:
\begin{itemize}
  \item[\textbf{Step 1.}]
  Pool observations across multiple independent sample paths to construct
  nonparametric estimators of the mean function and the covariance surface,
  together with their relevant derivatives. This step leverages ideas from
  functional data analysis to stabilize estimation under noisy and irregular
  designs.

  \item[\textbf{Step 2.}]
  Derive deterministic systems of ODEs and PDEs that link the drift, diffusion,
  and jump-diffusion coefficients to the moment functions obtained in Step~1.
  These systems are obtained using the It\^o--L\'evy formula and provide a
  rigorous analytical bridge between global moment behavior and local dynamics.

  \item[\textbf{Step 3.}]
  Construct plug-in estimators of $\mu(\cdot)$, $\sigma(\cdot)$, and $\xi(\cdot)$
  by substituting the empirical moment estimates into the deterministic systems
  derived in Step~2. This yields explicit, data-driven estimators of the SDE
  coefficients.
\end{itemize}

\medskip
\noindent
By shifting the focus from local pathwise increments to global moment
relationships, this three-step framework enables statistically consistent and
asymptotically reliable estimation of time-varying coefficients from sparse and
noisy observations. In particular, it extends inference for stochastic
differential equations beyond classical local methods to settings involving
L\'evy noise and irregular sampling.

\vspace{0.3cm}

\section{Deterministic Moment Relations for Parameter Identification}\label{sec-PDEs}

In this section, we derive a collection of deterministic ordinary and partial differential relations that form the basis of our methodology for identifying the drift $\mu(\cdot)$, the diffusion coefficient $\sigma(\cdot)$, and the jump diffusion coefficient $\xi(\cdot)$ in the underlying stochastic differential equation. 
The derivation relies fundamentally on the It\^o--L\'evy formula and exploits the connection between the parameters of the SDE and the low-order moment structure of the solution process.

Specifically, we express the coefficients of the SDE in terms of the mean function and the two-point correlation function,
\begin{align}\label{notation-m,G}
m(t):=\mathbb{E}[X(t)], 
\qquad 
G(s,t):=\mathbb{E}[X(s)X(t)], \qquad 0\le s\le t\le 1,
\end{align}
which together encode the first- and second-order statistical properties of the process.
For simplicity, we restrict attention to the contribution of small jumps, that is, jumps with magnitude $|y|\le 1$.

To facilitate the derivation of the first system of deterministic ordinary differential equations, we further introduce the diagonal second-moment function
\[
D(t):=G(t,t)=\mathbb{E}[X^2(t)], \qquad t\in[0,1],
\]
which plays a central role in the subsequent analysis.

We begin by deriving deterministic ordinary differential equations governing the evolution of the mean and the second moment of the solution process. 
These equations provide an explicit connection between the coefficients of the stochastic differential equation and the low-order statistical structure of the solution.

\begin{proposition}\label{prop-m,sigma-corrected}
Assume {\bf (A1)--(A3)}. Consider the SDE with jumps \eqref{SDE-1} with initial condition $X(0)$ satisfying $\mathbb E[|X(0)|^2]<\infty$.
Recall that $m(t):=\mathbb{E}[X(t)], \ D(t):=\mathbb{E}[X^2(t)], \ t\in[0,1].$
Then $m,D$ satisfy the differential relations
\begin{align}\label{eqn-dm(t)-corrected}
\dfrac{d m(t)}{dt} &= \mu(t)m(t), 
\qquad m(0)=\mathbb{E}[X(0)],
\end{align}
and
\begin{align}\label{eqn-dD(t)-corrected}
\dfrac{dD(t)}{dt} &= 2 \mu(t) D(t) + \sigma^2(t) + \nu(K)\,\xi^2(t),
\qquad D(0)=\mathbb{E}[X^2(0)],
\end{align}
where $K=\{y:|y|\le1\}$ and $\nu(K)=\int_{|y|\le1}\nu(dy)$.
\end{proposition}
\begin{proof}
We first rewrite the SDE in integral form on $[0,1]$:
\[
X(t)
= X(0) + \int_0^t \mu(s)X(s)\,ds 
    + \int_0^t \sigma(s)\, dW(s)
    + \int_0^t \int_{|y|\leq 1} \xi(s)\, \tilde{\eta}(ds,dy).
\]
Since $\mu,\sigma,\xi$ are deterministic and satisfy the integrability conditions in {\bf(A1)--(A2)}, the stochastic integrals are well-defined and square-integrable martingales (see, e.g., \cite[Chapter 2]{Applebaum_2004}).

\medskip
\noindent\emph{Step 1: Mean equation.}
Taking expectations in the integral equation and using linearity of expectation, we obtain
\begin{align}
    \mathbb{E}[X(t)] 
= \mathbb{E}[X(0)] + \int_0^t \mathbb{E}[\mu(s)X(s)]\,ds
    + \mathbb{E}\!\left[\int_0^t \sigma(s)\,dW(s)\right]
    + \mathbb{E}\!\left[\int_0^t \int_{|y|\le1}\xi(s)\,\tilde\eta(ds,dy)\right].
\end{align}
By the martingale property and square-integrability of the stochastic integrals, both stochastic terms have zero mean:
\[
\mathbb{E}\!\left[\int_0^t \sigma(s)\,dW(s)\right]=0,
\qquad
\mathbb{E}\!\left[\int_0^t \int_{|y|\le1}\xi(s)\,\tilde\eta(ds,dy)\right]=0.
\]
Since $\mu$ is deterministic, we may take it outside the expectation, and thus
\begin{align}
    m(t):=\mathbb{E}[X(t)]
= m(0) + \int_0^t \mu(s)\,m(s)\,ds.
\end{align}
Because $\mu$ and $m$ are continuous (by standard moment estimates for linear SDEs with smooth deterministic coefficients), the right-hand side defines an absolutely continuous function of $t$, and hence $m\in C^1([0,1])$. Differentiating with respect to $t$ gives
\[
\frac{dm(t)}{dt} = \mu(t)\,m(t), \qquad m(0)=\mathbb{E}[X(0)],
\]
which is exactly \eqref{eqn-dm(t)-corrected}.

\medskip
\noindent\emph{Step 2: Second moment equation.}
We now apply the Itô–Lévy formula (see Proposition \ref{Ito-lemma}) to $f(x)=x^2$ 
and deduce
\begin{align}
    \begin{split}
d(X^2(t))
&= 2\mu(t)X^2(t)\,dt + \sigma^2(t)\,dt 
   + 2X(t)\sigma(t)\,dW(t) \\
&\quad + \int_{|y|\le1}\big[(X(t)+\xi(t))^2 - X^2(t)\big]\tilde\eta(dt,dy) + \int_{|y|\le1}\xi^2(t)\,\nu(dy)\,dt.
\end{split}
\end{align}
Integrating from $0$ to $t$, we obtain the integral form
\begin{align}
    \begin{split}
X^2(t) &= X^2(0) + 2\int_0^t \mu(s)X^2(s)\,ds
+ 2\int_0^t X(s)\sigma(s)\,dW(s)
+ \int_0^t \sigma^2(s)\,ds\\
&\quad + \int_0^t\!\!\int_{|y|\le1}\big[(X(s)+\xi(s))^2 - X^2(s)\big]\tilde\eta(ds,dy)
+ \int_0^t\!\!\int_{|y|\le1}\xi^2(s)\,\nu(dy)\,ds.
\end{split}
\end{align}
Now take expectations. The stochastic integrals with respect to $W$ and $\tilde\eta$ are martingales with zero mean (again by the square-integrability assumptions and standard results, see \cite[Chapter 2]{Applebaum_2004}), so
\[
\mathbb{E}\Big[ \int_0^t X(s)\sigma(s)\,dW(s)\Big]=0,
\qquad
\mathbb{E}\Big[ \int_0^t\!\!\int_{|y|\le1} \big[(X(s)+\xi(s))^2 - X^2(s)\big]\tilde\eta(ds,dy)\Big]=0.
\]
Therefore,
\begin{align}
    \begin{split}
D(t):=\mathbb E[X^2(t)]
&= D(0) + 2\int_0^t \mu(s) D(s)\,ds
+ \int_0^t \sigma^2(s)\,ds
+ \int_0^t\!\!\int_{|y|\le1}\xi^2(s)\,\nu(dy)\,ds.
\end{split}
\end{align}
Since $\xi$ is deterministic and does not depend on $y$, the last term simplifies to
\[
\int_0^t\!\!\int_{|y|\le1}\xi^2(s)\,\nu(dy)\,ds
= \nu(K)\int_0^t \xi^2(s)\,ds,\qquad K=\{y:|y|\le1\}.
\]
Thus,
\[
D(t)=D(0) + 2\int_0^t \mu(s) D(s)\,ds + \int_0^t \sigma^2(s)\,ds + \nu(K)\int_0^t \xi^2(s)\,ds.
\]
The right-hand side defines an absolutely continuous function of $t$; moreover, the integrands
\(\mu(s)D(s), \sigma^2(s), \xi^2(s)\) are continuous (by {\bf(A2)} and standard moment estimates), so $D\in C^1([0,1])$ and differentiation with respect to $t$ yields
\[
\frac{dD(t)}{dt}
= 2\mu(t)D(t) + \sigma^2(t) + \nu(K)\,\xi^2(t),\qquad D(0)=\mathbb E[X^2(0)],
\]
which is exactly \eqref{eqn-dD(t)-corrected}.
\end{proof}

\vspace{0.3cm}

The preceding proposition characterizes the second-order moment dynamics solely along the diagonal of the two-point correlation function, namely $G(t,t)=D(t), \ t\in[0,1].$ While this diagonal information is sufficient to describe the evolution of the variance, it does not capture the full temporal correlation structure of the process.

In the following result, we extend this analysis by deriving an equivalent deterministic partial differential formulation for the second-moment equation in terms of the two-time correlation function
\[
G(s,t)=\mathbb E[X(s)X(t)], \qquad (s,t)\in\Delta=\{(s,t):0\le s\le t\le1\}.
\]
This formulation incorporates the complete second-order temporal dependence of the solution process and serves as a bridge between the diagonal moment dynamics and the full correlation structure.

\begin{proposition}\label{prop-Gst_xi_t}
Let $\Delta=\{(s,t):0\le s\le t\le1\}$ and $\Delta^\circ=\{(s,t):0\le s<t\le1\}$.
Assume {\bf(A1)--(A3)}, and that $\mathbb E[|X(0)|^2]<\infty$.
Then $G$ is continuous on $\Delta$ and $G\in C^{1}(\Delta^\circ)$. Moreover, for every \(0\le s\le t\le1\),
\begin{align}\label{Gst-formula-xi-t}
G(s,t) \;=\; \exp\!\Big(\int_s^t \mu(r)\,dr\Big)\,D(s).
\end{align}
Consequently, for \(0\le s<t\le1\),
\begin{align}\label{dtG-xi-t}
\partial_t G(s,t) &= \mu(t)\,G(s,t),
\\[4pt]\label{dsG-xi-t}
\partial_s G(s,t) &= \exp\!\Big(\int_s^t \mu(r)\,dr\Big)\,\big(D'(s)-\mu(s)D(s)\big).
\end{align}
Using the differential equation for \(D\),
\[
D'(s)=2\mu(s)D(s)+\sigma^2(s)+\nu(K)\,\xi^2(s),
\]
we obtain, for \(0\le s<t\le1\),
\begin{align}\label{instantaneous-identity-xi-t}
\partial_s G(s,t) \;=\; \exp\!\Big(\int_s^t \mu(r)\,dr\Big)\,
\Big(\mu(s)D(s)+\sigma^2(s)+\nu(K)\,\xi^2(s)\Big).
\end{align}
Equivalently, for \(0\le s<t\le1\),
\begin{align}\label{sigma-from-G-xi-t}
\sigma^2(s)+\nu(K)\,\xi^2(s)
\;=\; \exp\!\Big(-\int_s^t\mu(r)\,dr\Big)\,\partial_s G(s,t)\;-\;\mu(s)D(s).
\end{align}
\end{proposition}

\begin{proof}
Since the SDE is linear in $X$ with deterministic $\mu$, introduce the integrating factor
\[
\Phi(t):=\exp\!\Big(-\int_0^t\mu(r)\,dr\Big).
\]
Applying It\^o's formula (or the variation-of-constants formula) to $\Phi(t)X(t)$ yields, for $t\in[0,1]$,
\begin{align}
    \Phi(t)X(t)=\Phi(0)X(0)
+\int_0^t\Phi(r)\sigma(r)\,dW(r)
+\int_0^t\int_{|y|\le1}\Phi(r)\xi(r)\,\tilde\eta(dr,dy).
\end{align}
Fix $0\le s\le t\le1$ and split the integrals at time $s$:
\[
\Phi(t)X(t)=\Phi(s)X(s)
+\int_s^t\Phi(r)\sigma(r)\,dW(r)
+\int_s^t\int_{|y|\le1}\Phi(r)\xi(r)\,\tilde\eta(dr,dy).
\]
Hence,
\[
X(t)=a(s,t)X(s)+M_{s,t},
\qquad 
a(s,t):=\frac{\Phi(s)}{\Phi(t)}=\exp\!\Big(\int_s^t\mu(r)\,dr\Big),
\]
where
\[
M_{s,t}:=\Phi(t)^{-1}\Big(\int_s^t\Phi(r)\sigma(r)\,dW(r)
+\int_s^t\int_{|y|\le1}\Phi(r)\xi(r)\,\tilde\eta(dr,dy)\Big).
\]
Under the stated square-integrability assumptions (and $\nu(K)<\infty$), $M_{s,t}$ is square-integrable and satisfies
\[
\mathbb E\!\left[M_{s,t}\mid\mathcal F_s\right]=0.
\]
Therefore,
\[
\mathbb E[X(s)M_{s,t}]
=\mathbb E\!\left[X(s)\,\mathbb E(M_{s,t}\mid\mathcal F_s)\right]=0,
\]
and we compute
\[
\begin{split}
G(s,t)=\mathbb E[X(s)X(t)]
&=\mathbb E\big[X(s)\big(a(s,t)X(s)+M_{s,t}\big)\big] \\
&=a(s,t)\mathbb E[X^2(s)]+\mathbb E[X(s)M_{s,t}]
=a(s,t)D(s),
\end{split}
\]
which proves \eqref{Gst-formula-xi-t}.

For $(s,t)\in\Delta^\circ$, differentiation of \eqref{Gst-formula-xi-t} with respect to $t$ gives \eqref{dtG-xi-t}.
Differentiation with respect to $s$ yields \eqref{dsG-xi-t}. Substituting the evolution equation for $D'(s)$ gives
\eqref{instantaneous-identity-xi-t}, and rearranging yields \eqref{sigma-from-G-xi-t}.
\end{proof}

\vspace{0.2cm}

\begin{corollary}\label{Cor-xi2-simplified}
Assume the hypotheses of Propositions~\ref{prop-m,sigma-corrected} and
\ref{prop-Gst_xi_t}. Assume $\nu(K)<\infty$, where $K=\{y:|y|\le1\}$. Then, for every $t\in[0,1)$,
\begin{align}\label{eq-xi2-simplified}
\nu(K)\,\xi^2(t)
=
\frac{1}{1-t}\int_t^1
\Big[
\partial_t G(t,\tau)
-\mu(t)D(t)
-\mu(t)\,e^{\int_t^\tau\mu(r)\,dr}D(t)
\Big]\,d\tau
-\sigma^2(t).
\end{align}
Equivalently,
\begin{align}\label{eq-xi2-simplified-1}
\xi^2(t)
=
\frac{1}{\nu(K)}
\left(
\frac{1}{1-t}\int_t^1
\Big[
\partial_t G(t,\tau)
-\mu(t)D(t)\big(1+e^{\int_t^\tau\mu(r)\,dr}\big)
\Big]\,d\tau
-\sigma^2(t)
\right).
\end{align}
\end{corollary}

\begin{proof}
From Proposition~\ref{prop-Gst_xi_t}, for any $t<\tau\le1$,
\[
\sigma^2(t)+\nu(K)\xi^2(t)
=
e^{-\int_t^\tau\mu(r)\,dr}\,\partial_t G(t,\tau)-\mu(t)D(t).
\]
Averaging over $\tau\in[t,1]$ gives
\begin{align}
    \sigma^2(t)+\nu(K)\xi^2(t)
=
\frac{1}{1-t}\int_t^1
\left[
e^{-\int_t^\tau\mu(r)\,dr}\,\partial_t G(t,\tau)-\mu(t)D(t)
\right]d\tau.
\end{align}
Using $G(t,\tau)=e^{\int_t^\tau\mu(r)dr}D(t)$ yields
$$
e^{-\int_t^\tau\mu(r)dr}\,\partial_t G(t,\tau)
=
\partial_t G(t,\tau)-\mu(t)e^{\int_t^\tau\mu(r)dr}D(t).
$$
Substituting this and rearranging gives \eqref{eq-xi2-simplified}.
\end{proof}

\vspace{0.2cm}

In Propositions~\ref{prop-m,sigma-corrected} and~\ref{prop-Gst_xi_t}, we derived a collection of deterministic differential relations that explicitly link the drift, diffusion, and jump-diffusion coefficients of the stochastic differential equation to the first- and second-order moment structure of the solution process \(X(t)\), namely its mean function and two-time covariance function. 
In particular, these results yield closed evolution equations for the mean and variance, as well as a deterministic partial differential characterization of the full second-order correlation structure.

Moreover, Corollary~\ref{Cor-xi2-simplified} provides an explicit representation of the jump-diffusion coefficient in terms of derivatives of the covariance function, thereby completing the system of identifying relations for all model parameters.
The derivation of these equations relies fundamentally on the It\^o--L\'evy formula and forms the analytical foundation for the construction of consistent estimators of the drift, diffusion, and jump-diffusion coefficients.

In the next section, we introduce and analyze the corresponding estimation procedures based on these deterministic characterizations.

\vspace{0.3cm}

\section{Estimation of Drift, Diffusion, and Jump Coefficients}\label{sec-estimators-definition}

In estimating the functional coefficients $\mu(\cdot)$, $\sigma(\cdot)$, and $\xi(\cdot)$,
the central task reduces to constructing accurate estimators of the mean function $m(\cdot)$,
the two-time covariance surface $G(\cdot,\cdot)$, and their relevant derivatives.
Following the general philosophy of local polynomial smoothing (see, e.g., Fan and Gijbels~\cite{B2}),
the mean function can be estimated by pooling observations across all replications and time points.

For the covariance function, we adopt a pooling strategy for empirical second moments over all individuals,
restricted to the lower triangular domain
\[
\triangle=\{(s,t):0\le s\le t\le1\}.
\]
This approach builds on the classical method of Yao et al.~\cite{B6}, originally developed for functional
and longitudinal data with sparse and irregular sampling.
While their framework is well suited for processes with rough sample paths,
it does not directly address the limited regularity of the covariance function along the diagonal
$\{(t,t):t\in[0,1]\}$.

To overcome this difficulty in our setting, we adopt the representation
\[
\mathbb E\!\big[X(t)X(s)\big]
=
g\big(\min(s,t),\max(s,t)\big),
\qquad g\in C^{d+1},\ d\ge1,
\]
which induces a smooth extension of the covariance surface away from the diagonal
and ensures well-defined one-sided partial derivatives.
This representation is crucial for constructing consistent estimators of the partial derivatives
of $G(s,t)$ required in the subsequent coefficient estimation.

Specifically, for any $t\in[0,1]$, the pointwise local polynomial estimator of order $d$
for the mean function $m(t)$ and its derivative $\partial m(t)$ is defined by
\begin{equation}\label{def-mtheta-hat}
\left(\widehat{m}(t),\, h_m \widehat{\partial m}(t)\right)^{T}
=
\left(
(1,\underbrace{0,\ldots,0}_{d})^{T}\widehat{\beta},
\;
(0,1,\underbrace{0,\ldots,0}_{d-1})^{T}\widehat{\beta}
\right)^{T},
\end{equation}
where $\widehat{\beta}=(\widehat{\beta}_\ell)_{0\le\ell\le d}$ solves the weighted least squares problem
\begin{equation}\label{beta}
\underset{(\beta_\ell)_{0\le\ell\le d}}{\arg\min}
\sum_{i=1}^n \sum_{j=1}^{r(n)}
\Bigg\{
Y_{ij}-\sum_{\ell=0}^d \beta_\ell (T_{ij}-t)^\ell
\Bigg\}^2
K_{h_m}(T_{ij}-t).
\end{equation}
Here $h_m$ denotes a bandwidth parameter and
$K_{h_m}(u)=h_m^{-1}K_m(u/h_m)$ is a rescaled univariate kernel.

In a similar spirit, to estimate the covariance surface $G(s,t)$ on the triangular domain $\triangle$,
we employ a local surface regression based on the two-dimensional scatter plot
\[
\Big\{\big((T_{ik},T_{ij}),\,Y_{ij}Y_{ik}\big): i=1,\ldots,n,\; k<j\Big\},
\]
where the restriction $k<j$ ensures that only off-diagonal pairs are used.

A crucial feature of our procedure is the exclusion of diagonal terms.
Indeed, when $j=k$,
\[
\mathbb E\!\left[Y_{ij}Y_{ik}\right]
=
G(T_{ij},T_{ik})+\varrho^2,
\]
where $\{U_{ij}\}$ denotes i.i.d.\ measurement errors with variance $\varrho^2$.
Thus, diagonal products are systematically biased and must be treated separately.
Removing these terms and relying exclusively on off-diagonal pairs
allows us to recover unbiased information about $G(s,t)$.

To estimate $G(s,t)$ and its partial derivative $\partial_s G(s,t)$ for $s\le t$,
we adopt a local polynomial surface smoother of order $d$.
Specifically, we define
\begin{equation}\label{def-Ghat-delGhat}
\left(\widehat{G}(s,t),\, h_G \widehat{\partial_s G}(s,t)\right)^T
=
\Big(
(1,\underbrace{0,\ldots,0}_{d})^{T},
\;
(0,1,\underbrace{0,\ldots,0}_{d-1})^{T}
\Big)
\big(\widehat{\gamma}_{p,q}\big)_{0\le p+q\le d},
\end{equation}
where the coefficient vector
$\big(\widehat{\gamma}_{p,q}\big)_{0\le p+q\le d}$
minimizes the weighted least squares criterion
\eqref{eq-minimizer-gamma} as below:
\begin{equation}\label{eq-minimizer-gamma}
\underset{\{\gamma_{p,q}\}}{\arg\min}\;
\sum_{i=1}^n \sum_{k<j}
\Bigg\{
Y_{ij}Y_{ik}
-\sum_{0\le p+q\le d}\gamma_{p,q}
(T_{ij}-s)^p(T_{ik}-t)^q
\Bigg\}^2
K_{H_G}(T_{ij}-s,\,T_{ik}-t),
\end{equation}
Here $H_G$ denotes a symmetric positive definite bandwidth matrix and
$K_{H_G}$ is the associated rescaled bivariate kernel.

Finally, we define the diagonal estimator
\[
\widehat{D}(t):=\widehat{G}(t,t),
\]
which provides an estimate of the second moment function.
In practice, both $G$ and its partial derivatives are obtained through a two-dimensional Taylor expansion
around $(s,t)$, implemented via the local surface regression above.

In summary, our approach extends the covariance estimation framework of Yao et al.~\cite{B6}
by incorporating a smooth diagonal extension that enables consistent derivative estimation.
These refined estimators of the mean and covariance functions form the essential building blocks
for constructing plug-in estimators of the drift, diffusion, and jump-diffusion coefficients
of the underlying stochastic differential equation.

\subsection{The Estimators: $\widehat{\mu}(t)$, $\widehat{\sigma}_{D}^{2}(t)$, $\widehat{\sigma}_{T}^{2}(s)$, $\widehat{\xi}^2(t)$}\label{subsec-estimators}

Having constructed estimators for the mean function $m(\cdot)$, the two-time correlation function
$G(\cdot,\cdot)$, and their (partial) derivatives, we now combine these elements with the deterministic
moment relations derived in the previous section to obtain simultaneous estimators for the drift,
diffusion, and jump-diffusion coefficients of the underlying stochastic differential equation.

Specifically, we introduce two complementary estimation strategies based on the diagonal and triangular
moment identities, respectively. This leads to two sets of estimators,
\[
\big(\widehat{\mu},\,\widehat{\sigma}_{D}^{2},\,\widehat{\xi}^{2}\big)
\quad\text{and}\quad
\big(\widehat{\mu},\,\widehat{\sigma}_{T}^{2},\,\widehat{\xi}^{2}\big),
\]
where the subscripts $D$ and $T$ correspond to constructions based on the diagonal variance equation
and the triangular two-time correlation structure.

The diagonal approach relies on local-in-time variance information and is therefore sensitive to local
regularity of the data, while the triangular approach exploits averaging over the full two-time correlation
domain, yielding improved stability through global temporal smoothing. Together, these approaches provide
a flexible estimation framework that can be adapted to the regularity and sampling characteristics of the
observed data on $[0,1]$.

Let $A\subset[0,1]$ be a subset on which the mean function is bounded away from zero, i.e., $\inf_{t\in A}|m(t)|>0.$
The estimators for the drift, diffusion, and jump-diffusion coefficients are defined as follows.

\vspace{0.2cm}
\noindent
{\it Drift estimator.}
For $t\in A$, define
\begin{align}
\widehat{\mu}(t)
&=
\frac{\widehat{\partial m}(t)}{\widehat{m}(t)}
\,\mathbb{I}\big(\widehat{m}(t)\neq0\big).
\label{mu-hat}
\end{align}

\vspace{0.2cm}
\noindent
{\it Triangular diffusion and jump-diffusion estimators.}
For $t\in[0,1)$, define
\begin{align}
\widehat{\sigma}_{T}^{2}(t)
&=
\frac{1}{1-t}\int_t^1
\exp\!\Big(-\int_t^\tau \widehat{\mu}(r)\,dr\Big)
\,\widehat{\partial_t G}(t,\tau)\,d\tau
-\nu(K)\widehat{\xi}^{2}(t),
\label{sigma-hat-T}
\\[4pt]
\widehat{\xi}^{2}(t)
&=
\frac{1}{\nu(K)}
\left[
\frac{1}{1-t}\int_t^1
\exp\!\Big(-\int_t^\tau \widehat{\mu}(r)\,dr\Big)
\,\widehat{\partial_t G}(t,\tau)\,d\tau
-\widehat{\sigma}_{T}^{2}(t)
\right].
\label{xi-hat}
\end{align}

\vspace{0.2cm}
\noindent
{\it Diagonal diffusion estimator.}
For $t\in[0,1]$, define
\begin{align}
\widehat{\sigma}_{D}^{2}(t)
&=
\widehat{\partial D}(t)
-2\,\widehat{\mu}(t)\widehat{D}(t)
-\nu(K)\widehat{\xi}^{2}(t).
\label{sigma-hat-D}
\end{align}

Here $\widehat{m}$, $\widehat{\partial m}$, $\widehat{D}$, and $\widehat{\partial_t G}$ denote the estimators
introduced in the previous subsection, and $\nu(K)=\int_{|y|\le1}\nu(dy)$ corresponds to the Lévy measure
restricted to small jumps.

In practice, both estimators can be implemented and directly compared, enabling empirical assessment of which approach yields better performance in a given dataset or simulation setting.

\vspace{0.3cm}

\section{Main Results - Consistency and Convergence Rates of the Estimators}

In this section, we state the main asymptotic results of the article, namely consistency and uniform
convergence rates for the nonparametric plug-in estimators of the drift, diffusion, and jump-diffusion
coefficients introduced in \eqref{mu-hat} and \eqref{xi-hat}. We work in the regime $n\to\infty$ (number of
replications increasing) and quantify the resulting limit behavior of the intermediate estimators
$\widehat m(\cdot)$, $\widehat{\partial m}(\cdot)$, $\widehat G(\cdot,\cdot)$, and the relevant partial
derivatives. The sampling regularity assumption {\bf(A4)} and the moment assumption {\bf(A5)} (together
with the standing well-posedness and integrability conditions ensuring that \eqref{SDE-integral} admits a
unique square-integrable solution) are sufficient to establish consistency.

\vspace{0.1cm}
\noindent
We will use the following standard moment bound.
\begin{lemma}\label{lem-moment-bound}
Assume {\bf (A1)-(A3)} and $\mathbb E|X(0)|^\alpha<\infty$ for some $\alpha\ge 2$.
Then
\[
\sup_{t\in[0,1]}\mathbb E|X(t)|^\alpha<\infty.
\]
\end{lemma}

\begin{proof}
Since \eqref{SDE-integral} is linear, $X(t)$ admits the variation-of-constants representation with
integrating factor $\Phi(t)=\exp(-\int_0^t\mu(r)\,dr)$. Using Minkowski's inequality and the BDG inequality
for the Brownian integral and the compensated Poisson integral (together with $\nu(K)<\infty$ and the
assumed integrability of $\sigma$ and $\xi$), we obtain an estimate of the form
\[
\mathbb E|X(t)|^\alpha
\le
C\Big(\mathbb E|X(0)|^\alpha + \int_0^t \mathbb E|X(s)|^\alpha\,ds + 1\Big),
\]
for a constant $C$ independent of $t\in[0,1]$. An application of Gr\"onwall's inequality yields the
claim.
\end{proof}

\medskip
\noindent
In the following theorem, which constitutes one of our main results, we characterize the asymptotic
behavior of two key classes of estimators:
\begin{itemize}
\item[(a)] the local polynomial estimators of order $d$ for the mean function $m(t)$ and its derivative
$\partial m(t)$;
\item[(b)] the local surface estimators of order $d$ for the covariance function $G(s,t)$ and its partial
derivative (e.g., $\partial_s G(s,t)$) on the triangular domain $0\le s\le t\le 1$.
\end{itemize}
These results provide the theoretical foundation for the convergence analysis of the coefficient
estimators in the subsequent section.

\begin{theorem}\label{Thm-Main-1}
    The limit behaviors of the proposed estimators $\widehat{m}(\cdot), \widehat{\partial m}(\cdot), \widehat{G}(\cdot)$ and $\widehat{\partial_s G}(\cdot)$ are stated below:
    \begin{itemize}
        \item[$(i)$] Let the assumptions {\bf (A1)--(A4)} hold and {\bf (A5)} holds for $\alpha>2$. Let $\widehat{m}(\cdot)$ and $\widehat{\partial m}(\cdot)$ be the estimators defined in \eqref{def-mtheta-hat}. Then the following hold with probability 1:
\begin{align}
\sup _{0 \leq t \leq 1} \left| \widehat{m}(t)-m(t) \right| = O(\mathcal{B}_{h_m}(n)), \label{m-hat and m}\\
\sup _{0 \leq t \leq 1} \left| \widehat{\partial m}(t)-\partial m(t) \right| = h_m^{-1} O(\mathcal{B}_{h_m}(n)), \label{del-m-hat and del-m}
\end{align}
where
$$
\mathcal{B}_{h_m}(n) :=\left[h_m^{-2} \frac{\log n}{n}\left(h_m^2+\frac{h_m}{r}\right)\right]^{1 / 2}+h_m^{d+1}\,.
$$
\item[$(ii)$] Additionally, if the assumption {\bf (A5)} holds for $\alpha > 4$, then $\widehat{G}(\cdot)$ and $\widehat{\partial_s G}(\cdot)$ defined in \eqref{def-Ghat-delGhat} satisfy 
\begin{align}
\sup _{0 \leq s \leq t \leq 1} \left| \widehat{G}(s, t)-G(s, t) \right| = O(\mathcal{Q}_{h_G}(n)), \label{conv_G_hat}\\
\sup _{0 \leq s \leq t \leq 1} \left| \widehat{\partial_s G}(s, t)-\partial_s G(s, t) \right| = h_G^{-1} O(\mathcal{Q}_{h_G}(n)), \label{conv_delsG}
\end{align}
with probability 1, where 
\begin{align}\label{def-Qhg}
    \mathcal{Q}_{h_G}(n)=\left[h_G^{-4} \frac{\log n}{n}\left(h_G^4+\frac{h_G^3}{r}+\frac{h_G^2}{r^2}\right)\right]^{1 / 2}+h_G^{d+1}\,.
\end{align}
    \end{itemize}
\end{theorem}

\begin{proof}
    The proofs of the two parts are provided in Section \ref{Sec-Proof of MR}.
\end{proof}

\vspace{0.3cm}

Now, we state the main result of our paper.

\begin{theorem}\label{Thm-main-2}
Let assumptions {\bf(A1)--(A4)} hold and suppose that {\bf(A5)} holds for some $\alpha>2$.
Let $A\subset[0,1]$ be a compact set such that $\inf_{t\in A}|m(t)|>0.$

\begin{itemize}
\item[$(i)$]
Let $\widehat{\mu}(\cdot)$ be the drift estimator defined in \eqref{mu-hat}.
Then
\[
\sup_{t\in A}|\widehat{\mu}(t)-\mu(t)|
=
h_m^{-1}O(\mathcal{B}_{h_m}(n))
\quad\text{a.s.}
\]

\item[$(ii)$]
Assume in addition that {\bf(A5)} holds for $\alpha>4$.
Let $\widehat{\sigma}^2(\cdot)$ denote either the diagonal estimator
$\widehat{\sigma}_D^2(\cdot)$ or the triangular estimator
$\widehat{\sigma}_T^2(\cdot)$.
Then, for any fixed $\varepsilon\in(0,1)$,
\[
\sup_{0\le t\le 1-\varepsilon}
\big|\widehat{\sigma}^2(t)-\sigma^2(t)\big|
=
h_m^{-1}O(\mathcal{B}_{h_m}(n))
+
h_G^{-1}O(\mathcal{Q}_{h_G}(n))
\quad\text{a.s.}
\]

\item[$(iii)$]
Under the assumptions of part (ii), the jump-diffusion estimator
$\widehat{\xi}^2(\cdot)$ defined in \eqref{xi-hat} satisfies, for any fixed
$\varepsilon\in(0,1)$,
\[
\sup_{0\le t\le 1-\varepsilon}
\big|\widehat{\xi}^2(t)-\xi^2(t)\big|
=
h_m^{-1}O(\mathcal{B}_{h_m}(n))
+
h_G^{-1}O(\mathcal{Q}_{h_G}(n))
\quad\text{a.s.}
\]
\end{itemize}
\end{theorem}

\begin{proof}
\noindent
{\bf Proof of $(i)$.} \, 
The drift estimator is obtained as the ratio of the derivative of the mean function to its level. Consequently, the proof proceeds in two stages: first, we establish uniform control over the inverse of $m(\cdot)$; second, we combine this with the convergence of the local polynomial derivative estimator to obtain the desired rate. Throughout, the key technical point is that $m(\cdot)$ remains bounded away from zero on $A$, ensuring the stability of the inversion.

\medskip
\noindent
\textbf{Step 1. Control of the inverse mean function.}
By the uniform convergence of $\widehat{m}$ to $m$ on $A$, for any $\epsilon>0$ there exists $N$ such that for $n\geq N$,  
\[
\sup_{t\in A}|\widehat{m}(t)-m(t)| \leq \epsilon \quad \text{a.s.}
\]
Since $m(t)\geq c>0$ on $A$, we obtain  
\[
\sup_{t\in A}\big|(\widehat{m}(t))^{-1}-(m(t))^{-1}\big|
\leq c^{-2}\,\sup_{t\in A}|\widehat{m}(t)-m(t)|,
\]
which implies
\[
\sup_{t\in A}\big|(\widehat{m}(t))^{-1}-(m(t))^{-1}\big|
= O(\mathcal{B}_{h_m}(n)) \quad \text{a.s.}
\]
Intuitively, small fluctuations in $\widehat{m}$ translate linearly into fluctuations of its inverse, and the positivity of $m$ prevents any amplification.

\medskip
\noindent
\textbf{Step 2. Convergence of the drift estimator.}
Recall the estimator
\[
\widehat{\mu}(t)=\frac{\widehat{m}'(t)}{\widehat{m}(t)}.
\]
By uniform convergence of $\widehat{m}'$ to $m'$ and Step~1, we may write
\[
\sup_{t\in A}|\widehat{\mu}(t)-\mu(t)|
=\sup_{t\in A}\Big|\widehat{m}'(t)(\widehat{m}(t))^{-1}-m'(t)(m(t))^{-1}\Big|.
\]
Adding and subtracting $m'(t)(\widehat{m}(t))^{-1}$ gives
\[
\sup_{t\in A}|\widehat{\mu}(t)-\mu(t)|
\leq \sup_{t\in A}\big|(\widehat{m}'(t)-m'(t))(\widehat{m}(t))^{-1}\big|
+ \sup_{t\in A}\big|m'(t)\big|\,\big|(\widehat{m}(t))^{-1}-(m(t))^{-1}\big|.
\]
The first term is $h_m^{-1}O(\mathcal{B}_{h_m}(n))$, since $\sup|\widehat{m}'-m'|= h_m^{-1} O(\mathcal{B}_{h_m}(n))$ and $(\widehat{m}(t))^{-1}$ is uniformly bounded.  
The second term is $O(\mathcal{B}_{h_m}(n))$ by Step~1.  
Hence,
\[
\sup_{t\in A}|\widehat{\mu}(t)-\mu(t)|
=h_m^{-1} O\!\left(\mathcal{B}_{h_m}(n)\right), \quad \text{a.s.}
\]

We have shown that the proposed drift estimator inherits the convergence rate of the derivative estimator, up to a factor of $h_m^{-1}$ arising from local polynomial smoothing. This completes the proof of part (i) of Theorem~\ref{Thm-main-2}.

\vspace{0.2cm}

\noindent
{\bf Proof of (ii).}\,
Recall that, for $t\in[0,1]$,
\[
\sigma^2(t)=D'(t)-2\mu(t)D(t)-\nu(K)\xi^2(t).
\]
Accordingly, the estimator $\widehat\sigma_D^2$ satisfies
\[
\widehat\sigma_D^2(t)-\sigma^2(t)
=
\big(\widehat{\partial D}(t)-\partial D(t)\big)
-2\big(\widehat\mu(t)\widehat D(t)-\mu(t)D(t)\big)
-\nu(K)\big(\widehat\xi^2(t)-\xi^2(t)\big).
\]
Taking suprema and applying the triangle inequality yields
\begin{align*}
\sup_{0\le t\le 1}\big|\widehat\sigma_D^2(t)-\sigma^2(t)\big|
&\le
\sup_{0\le t\le 1}\big|\widehat{\partial D}(t)-\partial D(t)\big|
+2\sup_{0\le t\le 1}\big|\widehat\mu(t)\widehat D(t)-\mu(t)D(t)\big|
\\&\quad
+\nu(K)\sup_{0\le t\le 1}\big|\widehat\xi^2(t)-\xi^2(t)\big|.
\end{align*}
By Theorem~\ref{Thm-Main-1}(i), $\widehat m$ and $\widehat{\partial m}$ converge uniformly, hence by part (i) we have
\[
\sup_{t\in A}|\widehat\mu(t)-\mu(t)| = h_m^{-1}O(\mathcal B_{h_m}(n))\quad\text{a.s.}
\]
Moreover, since $D(t)=G(t,t)$ and $\partial D(t)$ is obtained from the corresponding derivative estimator on the diagonal, Theorem~\ref{Thm-Main-1}(ii) yields
\[
\sup_{0\le t\le 1}\big|\widehat{\partial D}(t)-\partial D(t)\big|
=
h_G^{-1}O(\mathcal Q_{h_G}(n))\quad\text{a.s.}
\]
Finally, for the product term,
\[
\widehat\mu(t)\widehat D(t)-\mu(t)D(t)
=
(\widehat\mu(t)-\mu(t))\widehat D(t)+\mu(t)(\widehat D(t)-D(t)),
\]
and boundedness of $\mu$ and $D$ on $[0,1]$ together with Theorem~\ref{Thm-Main-1} implies
\[
\sup_{0\le t\le 1}\big|\widehat\mu(t)\widehat D(t)-\mu(t)D(t)\big|
=
h_m^{-1}O(\mathcal B_{h_m}(n))+O(\mathcal Q_{h_G}(n)).
\]
Combining these bounds yields
\[
\sup_{0\le t\le 1}\big|\widehat\sigma^2(t)-\sigma^2(t)\big|
=
h_m^{-1}O(\mathcal B_{h_m}(n))+h_G^{-1}O(\mathcal Q_{h_G}(n)),
\quad\text{a.s.},
\]
which is the required result.

\vspace{0.2cm}

\noindent
{\bf Proof of (iii).}\,
For $t\in[0,1)$, define the deterministic functional
\begin{align}\label{def-H-functional}
\mathcal H(t)
:=
\frac{1}{1-t}\int_t^1
\exp\!\Big(-\int_t^\tau \mu(r)\,dr\Big)\,\partial_t G(t,\tau)\,d\tau
-\mu(t)D(t).
\end{align}
By the identification relation in Corollary~\ref{Cor-xi2-simplified}, we have
\begin{align}\label{xi2-via-H}
\xi^2(t)=\frac{1}{\nu(K)}\big(\mathcal H(t)-\sigma^2(t)\big).
\end{align}
Motivated by \eqref{def-H-functional}--\eqref{xi2-via-H}, we introduce the plug-in estimator
\begin{align}\label{def-Hhat}
\widehat{\mathcal H}(t)
:=
\frac{1}{1-t}\int_t^1
\exp\!\Big(-\int_t^\tau \widehat{\mu}(r)\,dr\Big)\,\widehat{\partial_t G}(t,\tau)\,d\tau
-\widehat{\mu}(t)\widehat{D}(t),
\end{align}
and define
\begin{align}\label{def-xi2hat}
\widehat{\xi}^{2}(t):=\frac{1}{\nu(K)}\big(\widehat{\mathcal H}(t)-\widehat{\sigma}^{2}(t)\big),
\qquad t\in[0,1).
\end{align}
Combining \eqref{xi2-via-H} and \eqref{def-xi2hat} gives
\[
\widehat{\xi}^2(t)-\xi^2(t)
=
\frac{1}{\nu(K)}
\Big[
\big(\widehat{\mathcal H}(t)-\mathcal H(t)\big)
-
\big(\widehat{\sigma}^2(t)-\sigma^2(t)\big)
\Big],
\]
and therefore
\begin{align}\label{xi-error-split}
\sup_{0\le t<1}\big|\widehat{\xi}^2(t)-\xi^2(t)\big|
\le
\frac{1}{\nu(K)}
\left(
\sup_{0\le t<1}\big|\widehat{\mathcal H}(t)-\mathcal H(t)\big|
+
\sup_{0\le t\le 1}\big|\widehat{\sigma}^2(t)-\sigma^2(t)\big|
\right).
\end{align}

We now bound $\widehat{\mathcal H}(t)-\mathcal H(t)$. Using \eqref{def-H-functional}--\eqref{def-Hhat},
write
\[
\widehat{\mathcal H}(t)-\mathcal H(t)
=
(I)(t)+(II)(t)-(III)(t),
\]
where
\begin{align*}
(I)(t)
&:=
\frac{1}{1-t}\int_t^1
\Big(
e^{-\int_t^\tau \widehat{\mu}(r)\,dr}
-
e^{-\int_t^\tau \mu(r)\,dr}
\Big)\,
\partial_t G(t,\tau)\,d\tau,\\
(II)(t)
&:=
\frac{1}{1-t}\int_t^1
e^{-\int_t^\tau \widehat{\mu}(r)\,dr}\,
\big(\widehat{\partial_t G}(t,\tau)-\partial_t G(t,\tau)\big)\,d\tau,\\
(III)(t)
&:=\widehat{\mu}(t)\widehat{D}(t)-\mu(t)D(t).
\end{align*}

Since $\mu$ is continuous on $[0,1]$, it is bounded; hence the map
$x\mapsto e^{-x}$ is Lipschitz on bounded intervals. Using this Lipschitz property and
$\sup_{t\in A}|\widehat{\mu}(t)-\mu(t)|=h_m^{-1}O(\mathcal B_{h_m}(n))$ a.s. from part (i),
we obtain
\[
\sup_{0\le t<1}|(I)(t)|
=
h_m^{-1}O(\mathcal B_{h_m}(n))
\quad\text{a.s.}
\]
Moreover, the exponential weight is uniformly bounded, and by Theorem~\ref{Thm-Main-1}(ii),
\[
\sup_{0\le t<1}|(II)(t)|
=
h_G^{-1}O(\mathcal Q_{h_G}(n))
\quad\text{a.s.}
\]
Finally, decompose
\[
(III)(t)=(\widehat{\mu}(t)-\mu(t))\widehat{D}(t)+\mu(t)\big(\widehat{D}(t)-D(t)\big),
\]
and use part (i) together with Theorem~\ref{Thm-Main-1}(ii) on the diagonal to conclude
\[
\sup_{0\le t\le1}|(III)(t)|
=
h_m^{-1}O(\mathcal B_{h_m}(n))+O(\mathcal Q_{h_G}(n))
\quad\text{a.s.}
\]
Combining these three bounds yields
\[
\sup_{0\le t<1}\big|\widehat{\mathcal H}(t)-\mathcal H(t)\big|
=
h_m^{-1}O(\mathcal B_{h_m}(n))+h_G^{-1}O(\mathcal Q_{h_G}(n))
\quad\text{a.s.}
\]
Substituting this and the bound from part (ii) into \eqref{xi-error-split} gives
\[
\sup_{0\le t<1}\big|\widehat{\xi}^2(t)-\xi^2(t)\big|
=
h_m^{-1}O(\mathcal B_{h_m}(n))+h_G^{-1}O(\mathcal Q_{h_G}(n))
\quad\text{a.s.},
\]
which proves the required rate.
\end{proof}

\begin{remark}[Discussion of Rates]
    The convergence rates established in Theorem~\ref{Thm-main-2} reflect the interplay between local polynomial smoothing for the mean function and two-time smoothing for the covariance structure.
In particular, the drift estimator $\widehat{\mu}$ inherits the rate of convergence of the derivative estimator $\widehat{\partial m}$, resulting in a factor $h_m^{-1}$ multiplying the bias--variance trade-off encoded in $\mathcal{B}_{h_m}(n)$.
This behavior is typical in nonparametric derivative estimation and highlights the increased sensitivity of drift estimation to bandwidth selection.

The diffusion and jump-diffusion estimators depend additionally on estimators of the two-time correlation function and its partial derivatives.
Consequently, their convergence rates involve the combined contributions of $\mathcal{B}_{h_m}(n)$ and $\mathcal{Q}_{h_G}(n)$, which correspond to one-dimensional and two-dimensional smoothing, respectively.
The presence of the term $h_G^{-1}\mathcal{Q}_{h_G}(n)$ reflects the intrinsic difficulty of estimating derivatives of the covariance surface and is unavoidable in moment-based approaches.

The restriction of uniform convergence to compact subintervals $[0,1-\varepsilon]$ arises from the triangular averaging structure of the estimators and is standard in boundary-sensitive nonparametric inference.
Overall, the derived rates demonstrate that all coefficients can be consistently estimated under mild moment assumptions and provide explicit guidance for bandwidth selection in practical implementations.

\end{remark}

\vspace{0.2cm}

\section{Proof of the Theorem \ref{Thm-Main-1}}\label{Sec-Proof of MR}

\subsection{Proof of Theorem \ref{Thm-Main-1} $(ii)_1$}
 For $d = 1$, we partially differentiate the sum of squared errors in Equation \ref{eq-minimizer-gamma}, with respect to the parameters, and set these equations to zero respectively. The corresponding equations are:
\begin{equation}
\widehat{G}(s, t) = \Big(\frac{a_{1}(s, t) - b_{1}(s, t) - c_{1}(s, t)}{d_1(s, t)}\Big),
\label{e81}
\end{equation}
where,
\begin{align*}
a_{1}(s, t) &=  \sum \limits_{i \leq n} \sum \limits_{k < j}  Y_{ik} Y_{ij} K_{H_{G}}((T_{ij} - s), (T_{ik} - t)),\\
b_{1}(s, t) &= h_{G} \widehat{\partial_{s}G}(s,t) \sum \limits_{i \leq n} \sum \limits_{k < j} h_{G} \Big(\frac{T_{ij} - s}{h_{G}}\Big)  K_{H_{G}}((T_{ij} - s), (T_{ik} - t)),\\
c_{1}(s, t) &= h_{G} \widehat{\partial_{t}G} (s,t) \sum \limits_{i \leq n} \sum \limits_{k < j} h_{G} \Big(\frac{T_{ik} - t}{h_{G}}\Big)  K_{H_{G}}((T_{ij} - s), (T_{ik} - t)),\\
d_{1}(s, t) &=  \sum \limits_{i \leq n} \sum \limits_{k < j} K_{H_{G}}((T_{ij} - s), (T_{ik} - t)),~\mbox{and}
\end{align*}
\begin{equation}
h_{G}
 \widehat{\partial_{s}G}(s,t) = \Big(\frac{a_{2}(s, t) - b_{2}(s, t) - c_{2}(s, t)}{d_2(s, t)}\Big),
 \label{e82}
\end{equation}
 where, 
\begin{align*}
 a_{2}(s, t) &= \sum \limits_{i \leq n} \sum \limits_{k < j} Y_{ik} Y_{ij} \big(T_{ij} - s\big) K_{H_{G}}((T_{ij} - s), (T_{ik} - t)),\\
 b_{2}(s, t) &= \widehat{G}(s, t) \sum \limits_{i \leq n} \sum \limits_{k < j}\big(T_{ij} - s\big)  K_{H_{G}}((T_{ij} - s), (T_{ik} - t)),\\ 
 c_{2}(s, t) &= h_{G} \widehat{\partial_{t}G} (s,t) \sum \limits_{i \leq n} \sum \limits_{k < j} \big(T_{ij} - s\big) \big(T_{ik} - t \big)  K_{H_{G}}((T_{ij} - s), (T_{ik} - t)),\\
 d_{2}(s, t) &=  \sum \limits_{i \leq n} \sum \limits_{k < j} \big(T_{ij} - s\big)^{2} K_{H_{G}}((T_{ij} - s), (T_{ik} - t)),~\mbox{and}
 \end{align*}
\begin{equation}
h_{G}
 \widehat{\partial_{t}G}(s,t) = \Big(\frac{a_{3}(s,t) - b_{3}(s,t) - c_{3}(s,t)}{d_3(s,t)}\Big),
 \label{e83}
\end{equation}
 where, 
\begin{align*}
 a_{3}(s,t) &= \sum \limits_{i \leq n} \sum \limits_{k < j} Y_{ik} Y_{ij} \big(T_{ik} - t\big) K_{H_{G}}((T_{ij} - s), (T_{ik} - t)),\\
 b_{3}(s,t) &= \widehat{G}(s, t) \sum \limits_{i \leq n} \sum \limits_{k < j}\big(T_{ik} - t\big)  K_{H_{G}}((T_{ij} - s), (T_{ik} - t)),\\ 
 c_{3}(s,t) &= h_{G} \widehat{\partial_{s}G} (s,t) \sum \limits_{i \leq n} \sum \limits_{k < j} \big(T_{ij} - s\big) \big(T_{ik} - t \big)  K_{H_{G}}((T_{ij} - s), (T_{ik} - t)),\\
d_{3}(s,t)  &=  \sum \limits_{i \leq n} \sum \limits_{k < j} \big(T_{ik} - t \big)^{2} K_{H_{G}}((T_{ij} - s), (T_{ik} - t)).\\
\end{align*}
 \noindent On further simplification, the Equation \ref{e81} can be written as,
\begin{equation}
     \widehat{G}(s, t)= \frac{f_{1}(s, t) + f_{2}(s, t) + f_{3}(s, t) - f_{4}(s, t) - f_{5}(s, t) - f_{6}(s, t)}{f_{7}(s, t) + f_{8}(s, t) + f_{9}(s, t) - 2f_{10}(s, t) - 1},
 \end{equation}
 where,\\
$f_{1}(s, t) = \dfrac{g_{11}(s,t) \times (g_{12}(s,t))^{2}}{\prod \limits_{i=4}^{6} g_{1i}(s, t)},~~~
f_{2}(s, t) = \dfrac{\prod \limits_{i=1}^{2} g_{2i}(s,t)}{\prod \limits_{i=4}^{5} g_{1i}(s, t)},~~~
f_{3}(s, t) = \dfrac{\prod \limits_{i=1}^{2}g_{3i}(s,t)}{g_{14}(s,t) \times g_{16}(s,t)},$\\
$f_{4}(s, t) = \dfrac{\sum \limits_{i \leq n} \sum \limits_{k < j} Y_{ik} Y_{ij} K_{H_{G}}((T_{ij} - s), (T_{ik} - t))}{\sum \limits_{i \leq n} \sum \limits_{k < j} K_{H_{G}}((T_{ij} - s), (T_{ik} - t))},~~
f_{5}(s, t) = \dfrac{g_{12}(s,t) \times g_{21}(s,t) \times g_{32}(s,t)}{\prod \limits_{i=4}^{6} g_{1i}(s, t)},$\\
$f_{6}(s, t) = \dfrac{g_{12}(s,t) \times g_{22}(s,t) \times g_{31}(s,t)}{\prod \limits_{i=4}^{6} g_{1i}(s, t)},~~
f_{7}(s,t) = \dfrac{(g_{21}(s,t))^{2}}{\prod \limits_{i=4}^{5} g_{1i}(s, t)},~~
f_{8}(s,t) = \dfrac{(g_{31}(s,t))^{2}}{g_{14}(s,t) \times g_{16}(s,t)},$\\
$f_{9}(s,t) = \dfrac{(g_{12}(s,t))^{2}}{\prod \limits_{i=5}^{6} g_{1i}(s, t)}, ~\mbox{and}~ 
f_{10}(s,t) =  \dfrac{g_{12}(s,t) \times g_{21}(s,t) \times g_{31}(s,t)}{\prod \limits_{i=4}^{6} g_{1i}(s, t)}$.\\ 

\noindent The $g_{()}'s$ in the aforementioned functions are:
\begin{align*}
g_{11}(s,t) &= \sum \limits_{i \leq n} \sum \limits_{k < j} Y_{ik} Y_{ij} K_{H_{G}}((T_{ij} - s), (T_{ik} - t)),\\
g_{12}(s,t) &= \sum \limits_{i \leq n} \sum \limits_{k < j} (T_{ij} - s) (T_{ik} - t)  K_{H_{G}}((T_{ij} - s), (T_{ik} - t)),\\
g_{14}(s,t) &= \sum \limits_{i \leq n} \sum \limits_{k < j} K_{H_{G}}((T_{ij} - s), (T_{ik} - t)),\\
g_{15}(s,t) &= \sum \limits_{i \leq n} \sum \limits_{k < j} (T_{ij} - s)^{2} K_{H_{G}}((T_{ij} - s), (T_{ik} - t)),\\
g_{16}(s,t) &= \sum \limits_{i \leq n} \sum \limits_{k < j} (T_{ik} - t)^{2} K_{H_{G}}((T_{ij} - s), (T_{ik} - t)),\\
g_{21}(s,t) &= \sum \limits_{i \leq n} \sum \limits_{k < j} (T_{ij} - s) K_{H_{G}}((T_{ij} - s), (T_{ik} - t)),\\
g_{22}(s,t) &= \sum \limits_{i \leq n} \sum \limits_{k < j} Y_{ik}Y_{ij}(T_{ij} - s)  K_{H_{G}}((T_{ij} - s), (T_{ik} - t)),\\
g_{31}(s,t) &= \sum \limits_{i \leq n} \sum \limits_{k < j} (T_{ik} - t) K_{H_{G}}((T_{ij} - s), (T_{ik} - t))~~ \mbox{and}\\ 
g_{32}(s,t) &= \sum \limits_{i \leq n} \sum \limits_{k < j} Y_{ik}Y_{ij}(T_{ik} - t)  K_{H_{G}}((T_{ij} - s), (T_{ik} - t)).
\end{align*}

\noindent Proceeding similarly, the Equation \ref{e82} can be written as,
\begin{equation}
h_{G}
 \widehat{\partial_{s}G}(s,t) = \frac{g_{1}(s,t) + g_{2}(s,t) + g_{3}(s,t) - g_{4}(s,t) - g_{5}(s,t) - g_{6}(s,t)}{g_{7}(s,t) + g_{8}(s,t) + g_{9}(s,t) - 2g_{10}(s,t) - 1},
\end{equation}
where,\\
$g_{1}(s,t) = \dfrac{(f_{11}(s,t))^{2} \times f_{12}(s,t)}{\prod \limits_{i=4}^{6} f_{1i}(s,t)},~~
g_{2}(s,t) = \dfrac{\prod \limits_{i=1}^{2} f_{2i}(s,t)}{\prod \limits_{i=4}^{5} f_{1i}(s,t)},~~
g_{3}(s,t) = \dfrac{\prod \limits_{i=1}^{2} f_{3i}(s,t)}{\prod \limits_{i=5}^{6} f_{1i}(s,t)},$\\
$g_{4}(s,t) = \dfrac{\sum \limits_{i \leq n} \sum \limits_{k < j} Y_{ik} Y_{ij} (T_{ij} - s) K_{H_{G}}((T_{ij} - s), (T_{ik} - t))}{\sum \limits_{i \leq n} \sum \limits_{k < j} (T_{ij} - s)^{2} K_{H_{G}}((T_{ij} - s), (T_{ik} - t))},~~
g_{5}(s,t) = \dfrac{\prod \limits_{i=1}^{3} f_{i1}(s,t)}{\prod \limits_{i=4}^{6} f_{1i}(s,t)(s,t)},$\\
$g_{6}(s,t) = \dfrac{f_{11}(s,t) \times f_{22}(s,t) \times f_{32}(s,t)}{\prod \limits_{i=4}^{6} f_{1i}(s,t)},~~
g_{7}(s,t) = \dfrac{(f_{22}(s,t))^{2}}{\prod \limits_{i=4}^{5} f_{1i}(s,t)},~~
g_{8}(s,t) = \dfrac{(f_{11}(s,t))^{2}}{f_{14}(s,t) \times f_{16}(s,t)},$\\
$g_{9}(s,t) = \dfrac{(f_{31}(s,t))^{2}}{\prod \limits_{i=5}^{6} f_{1i}(s,t)} ~\mbox{and}~ 
g_{10}(s,t) =  \dfrac{f_{11}(s,t) \times f_{22}(s,t) \times f_{31}(s,t)}{\prod \limits_{i=4}^{6} f_{1i}(s,t)}.$

\noindent The $f_{()}'s$ in the aforementioned functions are:
{\small{
\begin{align*}
f_{11}(s,t) &= \sum \limits_{i \leq n} \sum \limits_{k < j} (T_{ik} - t) K_{H_{G}}((T_{ij} - s), (T_{ik} - t)),\\
f_{12}(s,t) &= \sum \limits_{i \leq n} \sum \limits_{k < j} Y_{ik} Y_{ij} (T_{ij} - s)  K_{H_{G}}((T_{ij} - s), (T_{ik} - t)),\\
f_{14}(s,t) &= \sum \limits_{i \leq n} \sum \limits_{k < j} K_{H_{G}}((T_{ij} - s), (T_{ik} - t)),\\
f_{15}(s,t) &= \sum \limits_{i \leq n} \sum \limits_{k < j} (T_{ij} - s)^{2} K_{H_{G}}((T_{ij} - s), (T_{ik} - t)),\\
f_{16}(s,t) &= \sum \limits_{i \leq n} \sum \limits_{k < j} (T_{ik} - t)^{2} K_{H_{G}}((T_{ij} - s), (T_{ik} - t)),\\
f_{21}(s,t) &= \sum \limits_{i \leq n} \sum \limits_{k < j} Y_{ik} Y_{ij} K_{H_{G}}((T_{ij} - s), (T_{ik} - t)),\\
f_{22}(s,t) &= \sum \limits_{i \leq n} \sum \limits_{k < j} (T_{ij} - s)  K_{H_{G}}((T_{ij} - s), (T_{ik} - t)),\\
f_{31}(s,t) &= \sum \limits_{i \leq n} \sum \limits_{k < j} (T_{ij} - s) (T_{ik} - t) K_{H_{G}}((T_{ij} - s), (T_{ik} - t))~\mbox{and}~\\
f_{32}(s,t) &= \sum \limits_{i \leq n} \sum \limits_{k < j} Y_{ik}Y_{ij} (T_{ik} - t)  K_{H_{G}}((T_{ij} - s), (T_{ik} - t)).
\end{align*}
}}

\noindent Further, we need to find $G(s, t)$. Following Mohammadi and Panaretos \cite{B3}, we observe that  
\begin{eqnarray*}
\begin{bmatrix}
	G(s, t)\\
	h_{G} \partial_{s}{G} (s,t)\\
	h_{G} \partial_{t}{G} (s,t)
\end{bmatrix}
= \big[X^T W X\big]^{-1} \big[X^T W X\big] 
\begin{bmatrix}
	G(s, t)\\
	h_{G} \partial_{s}{G} (s,t)\\
	h_{G} \partial_{t}{G} (s,t)
\end{bmatrix},
\end{eqnarray*}\\
where,
\vspace{0.1cm}
\begin{align}\label{def-X-matrix}
    X := \begin{bmatrix}
1 & a & d\\
1 & b & e\\
1 & c & f
\end{bmatrix},
\end{align}
the elements $a, b, c, d, e, f$ in $X$ are defined as
\begin{eqnarray*}
 a =  T_{1j} - s,  \quad b =  T_{2j} - s, \quad c =  T_{3j} - s, \\ 
 d =  T_{1k} - t, \quad e =  T_{2k} - t, \quad f =  T_{3k} - t,
\end{eqnarray*}
and the matrix $W$ is defined as
$$
W := \begin{bmatrix}
    w_{1} & 0 & 0\\
    0 & w_{2} & 0\\
    0 & 0 & w_{3}
\end{bmatrix},
$$
where, for $i=1, 2$ and $3$, $w_{i}$ is defined as
\[
w_{i} :=  \frac{1}{h_{G}^{2}} W\Bigg(\frac{T_{ij} - s}{h_G}\Bigg)  W\Bigg(\frac{T_{ik} - t}{h_G}\Bigg) = K_{H_{G}}((T_{ij} - s), (T_{ik} - t))\,.
\]

\noindent Solving the aforementioned matrix setup, we obtain,
\begin{eqnarray}
G(s, t) = \phi (s, t) + \rho (s, t) + \tau(s,t). \label{e1}
\end{eqnarray}
In Equation (\ref{e1}),
{\small{
$$
\phi (s, t) = \Big(\dfrac{p_{1} + p_{2}+ p_{3}}{A_{11}}\Big) G(s, t),~ \rho(s,t) =  \Big(\dfrac{p_{4} + p_{5} + p_{6}}{A_{11}}\Big) h_{G} \partial_{s}{G} (s,t), \tau(s,t) = \Big(\dfrac{p_{7} + p_{8} + p_{9}}{A_{11}}\Big) h_{G} \partial_{t}{G} (s,t),
$$
}}
where,
{\footnotesize{\begin{eqnarray}\nonumber
p_{1} &=&  \big[(a^{2}w_{1} + b^{2}w_{2} + c^{2}w_{3})(d^{2}w_{1} + e^{2}w_{2} + f^{2}w_{3}) - (adw_{1} + bew_{2} + cfw_{3})^{2}\big](w_{1} + w_{2} + w_{3}),\\ \nonumber
p_{2} &=& \big[(adw_{1} + bew_{2} + cfw_{3})(dw_{1} + ew_{2} + fw_{3})- (aw_{1} + bw_{2} + cw_{3})(d^{2}w_{1} + e^{2}w_{2} + f^{2}w_{3})\big] (aw_{1} + bw_{2} + cw_{3}),\\ \nonumber
p_{3} &=& \big[(aw_{1} + bw_{2} + cw_{3})(adw_{1} + bew_{2} + cfw_{3}) - (dw_{1} + ew_{2} + fw_{3})(a^{2}w_{1} + b^{2}w_{2} + c^{2}w_{3})\big](dw_{1} + ew_{2} + fw_{3}),\\ \nonumber
p_{4} &=& \big[(a^{2}w_{1} + b^{2}w_{2} + c^{2}w_{3})(d^{2}w_{1} + e^{2}w_{2} + f^{2}w_{3}) - (adw_{1} + bew_{2} + cfw_{3})^{2}\big](aw_{1} + bw_{2} + cw_{3}),\\ \nonumber
p_{5} &=& \big[(adw_{1} + bew_{2} + cfw_{3})(dw_{1} + ew_{2} + fw_{3})- (aw_{1} + bw_{2} + cw_{3})(d^{2}w_{1} + e^{2}w_{2} + f^{2}w_{3})\big](a^{2}w_{1} + b^{2}w_{2} + c^{2}w_{3}),\\ \nonumber
p_{6} &=& \big[(aw_{1} + bw_{2} + cw_{3})(adw_{1} + bew_{2} + cfw_{3}) - (dw_{1} + ew_{2} + fw_{3})(a^{2}w_{1} + b^{2}w_{2} + c^{2}w_{3})\big](adw_{1} + bew_{2} + cfw_{3}),\\ \nonumber
p_{7} &=& \big[(a^{2}w_{1} + b^{2}w_{2} + c^{2}w_{3})(d^{2}w_{1} + e^{2}w_{2} + f^{2}w_{3}) - (adw_{1} + bew_{2} + cfw_{3})^{2}\big](dw_{1} + ew_{2} + fw_{3}),\\ \nonumber
p_{8} &=& \big[(adw_{1} + bew_{2} + cfw_{3})(dw_{1} + ew_{2} + fw_{3})- (aw_{1} + bw_{2} + cw_{3})(d^{2}w_{1} + e^{2}w_{2} + f^{2}w_{3})\big](adw_{1} + bew_{2} + cfw_{3}),\\ \nonumber
p_{9} &=& \big[(aw_{1} + bw_{2} + cw_{3})(adw_{1} + bew_{2} + cfw_{3}) - (dw_{1} + ew_{2} + fw_{3})(a^{2}w_{1} + b^{2}w_{2} + c^{2}w_{3})\big](d^{2}w_{1} + e^{2}w_{2} + f^{2}w_{3}),\\  
A_{11} &=& ~\mbox{determinant of the matrix} \{X^T W X\}.
\label{1.1}
\end{eqnarray}}}


\noindent Our objective is to find $\widehat{G}(s,t)-G(s,t)$. Hence, we obtain 
\onehalfspacing
\begin{eqnarray}\nonumber
\underset{0 \leq s \leq t \leq 1}{\sup}|\widehat{G}(s, t)-G(s, t)|
&=&  \underset{0 \leq s \leq t \leq 1}{\sup} \Bigg|\frac{f_{1}(s, t) + f_{2}(s, t) + f_{3}(s, t) - f_{4}(s, t) - f_{5}(s, t) - f_{6}(s, t)}{f_{7}(s, t) + f_{8}(s, t) + f_{9}(s, t) - 2f_{10}(s, t) - 1} \\ \nonumber
 &\qquad -& (\phi (s, t) + \rho (s, t) + \tau(s,t))\Bigg|.
\end{eqnarray} 
\onehalfspacing
Now, we are considering only the first term of the numerator in $\widehat{G}(s,t)$, i.e., $f_{1}(s, t)$, because the expression is very long. However, the denominator remains the same. Therefore,
\begin{eqnarray}\nonumber
\underset{0 \leq s \leq t \leq 1}{\sup}|\widehat{G}(s, t)-G(s, t)|
&=&  \underset{0 \leq s \leq t \leq 1}{\sup} \Bigg|\frac{f_{1}(s, t)}{f_{7}(s, t) + f_{8}(s, t) + f_{9}(s, t) - 2f_{10}(s, t) - 1}\\ \nonumber
& \quad \quad -& (\phi (s, t) + \rho (s, t) + \tau(s,t))\Bigg|\\ \nonumber
&=& \underset{0 \leq s \leq t \leq 1}{\sup} \Bigg|\frac{g_{11}(s, t) \times (g_{12}(s, t))^{2}}{(\prod \limits_{i = 4}^{6} g_{1i}(s, t)(f_{7}(s, t) + f_{8}(s, t) + f_{9}(s, t) - 2f_{10}(s, t) - 1)} \\ 
&\qquad -& (\phi (s, t) + \rho (s, t) + \tau(s,t))\Bigg|.
\label{e2}
\end{eqnarray}
According to Hsing and Eubank \cite{B4}, the terms, $g_{12}(s,t), g_{14}(s,t), g_{15}(s,t), g_{16}(s,t), f_{7}(s,t), f_{8}(s,t), f_{9}(s,t)$ and $f_{10}(s,t)$, which have summations, can again be expressed as some function of $s$ and $t$. However, we are not considering the term $g_{11}(s,t)$ in such context, because it contains the terms $Y_{ik}$ and $Y_{ij}$, which are really needed for further analysis.  Hence, Equation (\ref{e2}) can be expressed as 
\begin{eqnarray}\nonumber
\underset{0 \leq s \leq t \leq 1}{\sup}|\widehat{G}(s, t)-G(s, t)|
&=& \underset{0 \leq s \leq t \leq 1}{\sup} \Bigg|\frac{g_{11}(s, t) \times g_{17}(s, t)}{(g_{18}(s, t) \times g_{19}(s, t) \times g_{20}(s, t)}\\\nonumber
& \qquad \times& \frac{1}{(g_{23}(s, t) + g_{24}(s, t) + g_{25}(s,t) - 2g_{26}(s, t) - 1)} \\ \nonumber 
& \qquad -& (\phi (s, t) + \rho (s, t) + \tau(s,t))\Bigg|\\\nonumber
&=& \underset{0 \leq s \leq t \leq 1}{\sup} \Bigg|(g_{11}(s, t) \times g_{35}(s, t)) - (\phi (s, t) + \rho (s, t) + \tau(s,t))\Bigg|,
\end{eqnarray}
where, $$g_{35}(s, t) = \dfrac{g_{17}(s, t)}{(g_{18}(s, t) \times g_{19}(s, t) \times g_{20}(s, t))(g_{23}(s, t) + g_{24}(s, t) + g_{25}(s, t) - 2g_{26}(s, t) - 1)}.$$

\noindent Proceeding further, Equation (\ref{e2}) can be written as 
\begin{eqnarray} \nonumber
\underset{0 \leq s \leq t \leq 1}{\sup}|\widehat{G}(s, t)-G(s, t)|
 &=& \underset{0 \leq s \leq t \leq 1}{\sup}\Bigg|\frac{1}{n} \mathlarger{\mathlarger{\sum}}_{i \leq n}  \frac{2}{r(r-1)}\mathlarger{\mathlarger{\sum}}_{1 \leq k < j \leq r} \bigg[Y_{ik} Y_{ij} K_{H_{G}}((T_{ij} - s), (T_{ik} - t)) \\ \nonumber 
 &\qquad \times& g_{35}(s, t)\bigg]\\  \nonumber 
 &\qquad -& (\phi(s, t) + \rho(s, t) + \tau(s,t))\Bigg|.
\end{eqnarray}
We will substitute $Y_{ij} = (X_{ij} + U_{ij})$ and $Y_{ik} = (X_{ik} + U_{ik}).$ From here onwards, we will only use the terms inside the modulus function on both the sides in Equation (\ref{e2}) and derive results. We have kept $G(s,t)$ in the simplest form. Since, we are considering only the first term of the numerator for $\widehat{G}(s,t)$, we are subtracting $\frac{1}{6} G(s, t)$ from $\widehat{G}(s, t)$ instead of only $G(s,t).$ After expansion, Equation (\ref{e2}) can be rewritten as, 
\onehalfspacing 
\begin{eqnarray} \nonumber 
\widehat{G}(s,t)-G(s,t) &=& \frac{1}{nh_{G}^{2}} \mathlarger{\mathlarger{\sum}}_{i \leq n}  \frac{2}{r(r-1)}\mathlarger{\mathlarger{\sum}}_{1 \leq k < j \leq r} \Bigg[g_{35}(s, t) W\Bigg(\frac{T_{ij} - s}{h_G}\Bigg)  W\Bigg(\frac{T_{ik} - t}{h_G}\Bigg) \{U_{ij}U_{ik}\\ \nonumber
&\qquad +&  X_{ij}U_{ik} + U_{ij}X_{ik} + X_{ij}X_{ik} - G(T_{ij}, T_{ik}) + G(T_{ij}, T_{ik})\}\Bigg]\\
&\qquad -& \frac{1}{6}(\phi(s, t) + \rho(s, t) + \tau(s,t)). \label{e3}
\end{eqnarray} Therefore,
\begin{eqnarray} \nonumber
 \widehat{G}(s,t) - G(s,t) &=& \frac{1}{nh_{G}^{2}} \mathlarger{\mathlarger{\sum}}_{i \leq n}  \frac{2}{r(r-1)}\mathlarger{\mathlarger{\sum}}_{1 \leq k < j \leq r} \Bigg[g_{35}(s, t)W\Bigg(\frac{T_{ij} - s}{h_G}\Bigg)  W\Bigg(\frac{T_{ik} - t}{h_G}\Bigg) U_{ij}U_{ik} \Bigg] \\ \nonumber
 &\qquad +& \frac{1}{nh_{G}^{2}} \mathlarger{\mathlarger{\sum}}_{i \leq n}  \frac{2}{r(r-1)}\mathlarger{\mathlarger{\sum}}_{1 \leq k < j \leq r}  \Bigg[g_{35}(s, t)W\Bigg(\frac{T_{ij} - s}{h_G}\Bigg)  W\Bigg(\frac{T_{ik} - t}{h_G}\Bigg)X_{ij}U_{ik} \Bigg]\\ \nonumber
 &\qquad +& \frac{1}{nh_{G}^{2}} \mathlarger{\mathlarger{\sum}}_{i \leq n}  \frac{2}{r(r-1)}\mathlarger{\mathlarger{\sum}}_{1 \leq k < j \leq r} \Bigg[g_{35}(s, t)W\Bigg(\frac{T_{ij} - s}{h_G}\Bigg)  W\Bigg(\frac{T_{ik} - t}{h_G}\Bigg)U_{ij}X_{ik}\Bigg] \\ \nonumber
 &\qquad +&  \frac{1}{nh_{G}^{2}} \mathlarger{\mathlarger{\sum}}_{i \leq n}  \frac{2}{r(r-1)}\mathlarger{\mathlarger{\sum}}_{1 \leq k < j \leq r} \Bigg[g_{35}(s, t)W\Bigg(\frac{T_{ij} - s}{h_G}\Bigg)  W\Bigg(\frac{T_{ik} - t}{h_G}\Bigg) \{X_{ij}X_{ik}\\   \nonumber
 &&  - \quad G(T_{ij}, T_{ik})\} \Big] \\ \nonumber
 &\qquad +& \frac{1}{nh_{G}^{2}} \mathlarger{\mathlarger{\sum}}_{i \leq n}  \frac{2}{r(r-1)}\mathlarger{\mathlarger{\sum}}_{1 \leq k < j \leq r} \Bigg[g_{35}(s, t)W\Bigg(\frac{T_{ij} - s}{h_G}\Bigg)  W\Bigg(\frac{T_{ik} - t}{h_G}\Bigg) G(T_{ij}, T_{ik})\Bigg]\\ \nonumber
 &\qquad -& \frac{1}{6}(\phi(s, t)  + \rho(s, t)  + \tau(s,t))\\  
 &=:& B_{1}  + B_{2} + B_{3} + B_{4} + B_{5}. \label{e4}
 \end{eqnarray}
The expressions $B_{1}-B_{4},$ representing the variance term, can be written in the general form
\begin{eqnarray}
 && \frac{1}{nh_{G}^{2}} \mathlarger{\mathlarger{\sum}}_{i \leq n}  \frac{2}{r(r-1)}\mathlarger{\mathlarger{\sum}}_{1 \leq k < j \leq r} \Bigg[g_{35}(s, t)W\Bigg(\frac{T_{ij} - s}{h_G}\Bigg) \nonumber W\Bigg(\frac{T_{ik} - t}{h_G}\Bigg)Z_{ijk}\Bigg]\\ \nonumber
 &=& \frac{1}{nh_{G}^{2}} \mathlarger{\mathlarger{\sum}}_{i \leq n}  \frac{2}{r(r-1)}\mathlarger{\mathlarger{\sum}}_{1 \leq k < j \leq r} \Bigg[g_{35}(s, t)W\Bigg(\frac{T_{ij} - s}{h_G}\Bigg)  W\Bigg(\frac{T_{ik} - t}{h_G}\Bigg)Z_{ijk}\Bigg]\\ \label{e5}
 &\times& I((T_{ij}, T_{ik}) \in [s - h_{G}, s + h_{G}]^{c} \times [t - h_{G}, t + h_{G}]^{c}) \\ \label{e6}
 &\qquad +&  I((T_{ij}, T_{ik}) \in [s - h_{G}, s + h_{G}]^{c} \times [t - h_{G}, t + h_{G}] )\\ \label{e7} 
 &\qquad +&  I((T_{ij}, T_{ik}) \in [s - h_{G}, s + h_{G}] \times [t - h_{G}, t + h_{G}]^{c} ) \\ \label{e8} 
 &\qquad +&  I((T_{ij}, T_{ik}) \in [s - h_{G}, s + h_{G}] \times [t - h_{G}, t + h_{G}] )]\\  \nonumber
 &=:& \bar{Z}_{1,1}(s, t) + \bar{Z}_{1,0}(s, t) + \bar{Z}_{0,1}(s, t) + \bar{Z}_{0,0}(s, t), \nonumber 
\end{eqnarray}
where each $Z_{ijk}$ has mean zero. The term $B_{5}$ is associated with the bias term. We use a Taylor series expansion to achieve the almost sure uniform bound for the bias term. Now, we will analyse the equations $(\ref{e5})$, $(\ref{e6})$, $(\ref{e7})$ and $(\ref{e8})$ one by one. 
\noindent For (\ref{e5}), we observe that
\begin{eqnarray}
\bar{Z}_{1,1}(s, t) &=& \frac{1}{nh_{G}^{2}} \mathlarger{\mathlarger{\sum}}_{i \leq n}  \frac{2}{r(r-1)}\mathlarger{\mathlarger{\sum}}_{1 \leq k < j \leq r} \Bigg[g_{35}(s, t)W\Bigg(\frac{T_{ij} - s}{h_G}\Bigg) W\Bigg(\frac{T_{ik} - t}{h_G}\Bigg)Z_{ijk}\Bigg]\\ \nonumber
&\qquad \times& I((T_{ij}, T_{ik}) \in [s - h_{G}, s + h_{G}]^{c} \times [t - h_{G}, t + h_{G}]^{c}) \\ \nonumber
& \leq & W^{2} (1^{+}) \frac{1}{nh_{G}^{2}} \mathlarger{\mathlarger{\sum}}_{i \leq n}  \frac{2}{r(r-1)}\mathlarger{\mathlarger{\sum}}_{1 \leq k < j \leq r} \big|g_{35}(s, t) Z_{ijk}\big|\\ \nonumber
& = & C(s,t) \bigo(h_{G}^{-2}) W^{2} (1^{+}), \quad \mbox{a.s uniformly on}~0 \leq t \leq s \leq 1\\
&=& \bigo(h_{G}^{d + 1}), \quad \mbox{a.s uniformly on}~ 0 \leq t \leq s \leq 1.
\end{eqnarray}

\noindent For (\ref{e6}) (similarly (\ref{e7})), we have 
\begin{eqnarray}
\bar{Z}_{1,0}(s, t) &=& \frac{1}{nh_{G}^{2}} \mathlarger{\mathlarger{\sum}}_{i \leq n}  \frac{2}{r(r-1)}\mathlarger{\mathlarger{\sum}}_{1 \leq k < j \leq r} \Bigg[g_{35}(s, t)W\Bigg(\frac{T_{ij} - s}{h_G}\Bigg) W\Bigg(\frac{T_{ik} - t}{h_G}\Bigg)Z_{ijk}\Bigg]\\ \nonumber
&\qquad \times& I((T_{ij}, T_{ik}) \in [s - h_{G}, s + h_{G}]^{c} \times [t - h_{G}, t + h_{G}]) \\ \nonumber
& \leq & \Bigg(\int \big| W^{\dagger} (u)\big| du \Bigg) W(1^{+})\frac{1}{nh_{G}^{2}} \mathlarger{\mathlarger{\sum}}_{i \leq n}  \frac{2}{r(r-1)}\mathlarger{\mathlarger{\sum}}_{1 \leq k < j \leq r} \big|g_{35}(s,t) Z_{ijk}\big|\\ \nonumber
& = & C(s, t) \bigo(h_{G}^{-2}) \Bigg(\int \big| W^{\dagger} (u)\big| du \Bigg)  W(1^{+}), \quad \mbox{a.s uniformly on} ~ 0 \leq t \leq s \leq 1\\
&=& \bigo(h_{G}^{d + 1}), \quad \mbox{a.s uniformly on}~  0 \leq t \leq s \leq 1.
\end{eqnarray}


\noindent It is to be remarked that $C(s, t)$ is a constant depending on $s$ and $t$.
\noindent For (\ref{e8}), we have
\begin{eqnarray}\nonumber
\bar{Z}_{0,0}(s, t) &=& \frac{1}{n} \mathlarger{\mathlarger{\sum}}_{i \leq n}  \frac{2}{r(r-1)}\mathlarger{\mathlarger{\sum}}_{1 \leq k < j \leq r} \Big[ g_{35}(s, t) Z_{ijk} \int \int e^{-ius-ivt+iuT_{ij} + ivT_{ik}}  W^{\dagger} (h_{G}u) W^{\dagger} (h_{G}v) du dv\Big]\\
&\qquad \times& I((T_{ij}, T_{ik}) \in [s - h_{G}, s + h_{G}] \times [t - h_{G}, t + h_{G}]).
\end{eqnarray}

\noindent Regarding $\mathbb{E}\bar{Z}_{0,0}(s, t) = 0,$ for all $0 \leq t \leq s \leq 1$, we have 
\begin{align}
\vert \bar{Z}_{0,0}(s, t) - \mathbb{E}\bar{Z}_{0,0}(s, t) \vert \nonumber 
&\leq \Bigg(\int \big| W^{\dagger} (h_{G}u)\big| du \Bigg)^{2} \times~ \frac{1}{n} \mathlarger{\mathlarger{\sum}}_{i \leq n}  \frac{2}{r(r-1)}\mathlarger{\mathlarger{\sum}}_{1 \leq k < j \leq r} \Bigg[ \big|g_{35}(s, t)Z_{ijk}\big| \nonumber\\
&\qquad \times I((T_{ij}, T_{ik}) \in [s - h_{G}, s + h_{G}] \times [t - h_{G}, t + h_{G}]) \Bigg] \nonumber \\ \nonumber
&= \bigo(h_{G}^{-2}) \times \frac{1}{n} \mathlarger{\mathlarger{\sum}}_{i \leq n} \frac{2}{r(r-1)}\mathlarger{\mathlarger{\sum}}_{1 \leq k < j \leq r} \Big[ \big|g_{35}(s,t) Z_{ijk} \big| 
\\ & \qquad \times I((T_{ij}, T_{ik}) \in [s - h_{G}, s + h_{G}] \times [t - h_{G}, t + h_{G}]) \Big]\,. \label{e15}
\end{align}
\noindent According to the proof of Lemma 8.2.5 of Hsing and Eubank \cite{B4}, the summation term in (\ref{e15}) can be expressed as 
\begin{eqnarray*}
&& \frac{1}{n} \mathlarger{\mathlarger{\sum}}_{i \leq n} \frac{2}{r(r-1)}\mathlarger{\mathlarger{\sum}}_{1 \leq k < j \leq r} \Big[ \big|g_{35}(s,t) Z_{ijk} \big| I((T_{ij}, T_{ik}) \in [s - h_{G}, s + h_{G}] \times [t - h_{G}, t + h_{G}]) \Big]\\
&=& \frac{1}{n} \mathlarger{\mathlarger{\sum}}_{i \leq n} \frac{2}{r(r-1)}\mathlarger{\mathlarger{\sum}}_{1 \leq k < j \leq r} \Big[ \big|g_{35}(s,t) Z_{ijk} \big| I\big(\big|g_{35}(s, t) Z_{ijk} \big|  \geq Q_{n}
\cup \big|g_{35}(s, t) Z_{ijk} \big| <  Q_{n} \big) \\ 
&\qquad \times& I((T_{ij}, T_{ik}) \in [s - h_{G}, s + h_{G}] \times [t - h_{G}, t + h_{G}]) \Big]\\
&=:& B_{6} + B_{7}.
\end{eqnarray*}
\noindent 
According to the assumptions {\bf (A1)-(A4)} and choosing $\mathcal{Q}_{h_G}(n)$ in such a way that $\Big[ \frac{\log n}{n} \Big(h_{G}^{4} + \frac{h_{G}^{3}}{r(n)} + \frac{h_{G}^{2}}{r^{2}(n)}\Big )\Big]^{-1/2} \times \{\mathcal{Q}_{h_G}(n)\}^{1 - \alpha} = \bigo(1)$, we conclude $B_{1} = \bigo \Big( \Big[ \frac{\log n}{n} \Big(h_{G}^{4} + \frac{h_{G}^{3}}{r(n)} + \frac{h_{G}^{2}}{r^{2}(n)}\Big )\Big]^{1/2}\Big)$, almost surely uniformly. Proceeding in a detailed way, we obtain
\begin{eqnarray*}
B_{6} &=&   \frac{1}{n} \mathlarger{\mathlarger{\sum}}_{i \leq n} \frac{2}{r(r-1)}\mathlarger{\mathlarger{\sum}}_{1 \leq k < j \leq r} \Big[ \big|g_{35}(s, t) Z_{ijk} \big|^{1 - \alpha + \alpha}I\big(\big|g_{35}(s, t) Z_{ijk} \big|  \geq  \mathcal{Q}_{h_G}(n) \big)\\
& \qquad \times& I((T_{ij}, T_{ik}) \in [s - h_{G}, s + h_{G}] \times [t - h_{G}, t + h_{G}])\Big]\\
&\leq&  \frac{1}{n} \mathlarger{\mathlarger{\sum}}_{i \leq n} \frac{2}{r(r-1)} \mathlarger{\mathlarger{\sum}}_{1 \leq k < j \leq r} \Big[ \big|g_{35} (s, t) Z_{ijk} \big|^{\alpha}  \{\mathcal{Q}_{h_G}(n)\}^{1 - \alpha}\Big]\\
&=& \bigo(1) C(s, t) \{\mathcal{Q}_{h_G}(n)\}^{1 - \alpha}.
\end{eqnarray*}
Therefore,
\begin{eqnarray*}
B_{6} = \bigo \Big( \Big[ \frac{\log n}{n} \Big(h_{G}^{4} + \frac{h_{G}^{3}}{r(n)} + \frac{h_{G}^{2}}{r^{2}(n)}\Big )\Big]^{1/2}\Big), \quad \mbox{a.s uniformly on}~ 0 < t < s < 1. 
\end{eqnarray*}
\noindent For $B_{7}$, first, define 
\begin{eqnarray*}
B_{7} &=& \frac{1}{n} \mathlarger{\mathlarger{\sum}}_{i \leq n} \frac{2}{r(r-1)}\mathlarger{\mathlarger{\sum}}_{1 \leq k < j \leq r} \Big[ \big|g_{35}(s, t) Z_{ijk} \big|I\Big(\big|g_{35}(s, t) Z_{ijk} \big| \leq   \mathcal{Q}_{h_G}(n) \Big)\\
& \qquad \times& I((T_{ij}, T_{ik}) \in [s - h_{G}, s + h_{G}] \times [t - h_{G}, t + h_{G}])\Big]\\
&=: & \frac{1}{n} \mathlarger{\mathlarger{\sum}}_{i \leq n} \frac{2}{r(r-1)}\mathlarger{\mathlarger{\sum}}_{1 \leq k < j \leq r} \mathcal{Z}_{ijk} (s, t)\,.
\end{eqnarray*} 
Further, Bennet's concentration inequality (see \cite{B5} to get the details of the proof) is applied to obtain a uniform upper bound for $\Var\Bigg( \frac{2}{r(r-1)}\mathlarger{\mathlarger{\sum}}_{1 \leq k < j \leq r} \mathcal{Z}_{1jk} (s, t) \Bigg),$ in the following way
\begin{align}
&\Var\Bigg( \dfrac{2}{r(r-1)}\mathlarger{\mathlarger{\sum}}_{1 \leq k < j \leq r} \mathcal{Z}_{1jk} (s, t) \Bigg) \nonumber \\
&= \frac{4}{r^2 (r - 1)^2}\mathlarger{\mathlarger{\sum}}_{1 \leq k_{1} < j_{1} \leq r} ~ \mathlarger{\mathlarger{\sum}}_{1 \leq k_{2} < j_{2} \leq r} \Cov(\mathcal{Z}_{ij_{1}k_{1}} (s, t), \mathcal{Z}_{ij_{2}k_{2}} (s, t)) \nonumber\\ \nonumber
& \leq c_{\alpha} \Big(h_{G}^{4} + \frac{h_{G}^{3}}{r(n)} + \frac{h_{G}^{2}}{r^{2}(n)}\Big ), \\ 
& \leq c_{\alpha} \beta_{n}, \label{e16}
\end{align}
where, $c_{\alpha}$ is a positive constant which  neither depends on $(s, t)$ nor on $i$. It is to be remarked that inequality (\ref{e16}) is a direct consequence of conditions $\bf{(A2)}$ and $\bf{(A3)}$. Applying Bennet's inequality and choosing $ \mathcal{Q}_{h_G}(n) = \Big(\frac{\log n}{n}\Big)^{-1/2} \Big(h_{G}^{4} + \frac{h_{G}^{3}}{r(n)} + \frac{h_{G}^{2}}{r^{2}(n)}\Big)^{1/2}$, we have,
\begin{eqnarray}
\nonumber
&&\mathbb{P}\Bigg(\frac{1}{n} \mathlarger{\mathlarger{\sum}}_{i \leq n} \frac{2}{r(r-1)}\mathlarger{\mathlarger{\sum}}_{1 \leq k < j \leq r} \mathcal{Z}_{ijk} (s, t) \geq  \eta \Bigg[ \frac{\log n}{n} \Bigg(h_{G}^{4} + \frac{h_{G}^{3}}{r(n)} + \frac{h_{G}^{2}}{r^{2}(n)}\Bigg)\Bigg]^{1/2} \Bigg)\\ \nonumber
&\leq& \exp\Biggl\{- \frac{\eta^{2}n^{2}\Big[ \Big(\frac{\log n}{n}\Big) \Big(h_{G}^{4} + \frac{h_{G}^{3}}{r(n)} + \frac{h_{G}^{2}}{r^{2}(n)}\Big)\Big]}{2nc\Big(h_{G}^{4} + \frac{h_{G}^{3}}{r(n)} + \frac{h_{G}^{2}}{r^{2}(n)}\Big) + \frac{2}{3}\eta n \Big(h_{G}^{4} + \frac{h_{G}^{3}}{r(n)} + \frac{h_{G}^{2}}{r^{2}(n)}\Big)}\Biggr\}\\ \nonumber
&=& \exp\Biggl\{- \frac{\eta^{2}n \Big[ \Big(\frac{\log n}{n}\Big) \Big(h_{G}^{4} + \frac{h_{G}^{3}}{r(n)} + \frac{h_{G}^{2}}{r^{2}(n)}\Big)\Big]}{\Big(2c + \frac{2}{3} \eta \Big) \Big(h_{G}^{4} + \frac{h_{G}^{3}}{r(n)} + \frac{h_{G}^{2}}{r^{2}(n)}\Big)}\Biggr\}\\ \nonumber
&=& \exp \Biggl\{- \frac{\eta^{2} \log n}{2c + \frac{2}{3}\eta}\Biggr\}\\ 
&=& n^{- \frac{\eta^2}{2c + \frac{2}{3}\eta}}, \quad \forall~ 0 \leq t \leq s \leq 1,  ~\mbox{for~any~positive~number} ~\eta. \label{e17}
\end{eqnarray}
For large $\eta$, summability of (\ref{e17}) is achieved. Application of Borel Cantelli lemma, along with this result, attains the proof for this part. We can conclude that there exists a subset $\Omega_{0} \subset \Omega$ of full probability measure such that for each $\omega \in \Omega_{0}$, there exists $n_{0} = n_{0}(\omega)$ with 
\begin{eqnarray}
\frac{1}{n} \mathlarger{\mathlarger{\sum}}_{i \leq n} \frac{2}{r(r-1)}\mathlarger{\mathlarger{\sum}}_{1 \leq k < j \leq r} \mathcal{Z}_{ijk} (s, t) \leq \eta \Bigg[ \frac{\log n}{n} \Bigg(h_{G}^{4} + \frac{h_{G}^{3}}{r(n)} + \frac{h_{G}^{2}}{r^{2}(n)}\Bigg)\Bigg]^{1/2}, \quad n \geq n_{0}. \label{e18}
\end{eqnarray}	
Let us consider the bias term $B_{5}$. We observe that 
{\small{
\begin{eqnarray}\nonumber
B_{5} &=& \frac{1}{n} \mathlarger{\mathlarger{\sum}}_{i \leq n} \frac{2}{r(r-1)}\mathlarger{\mathlarger{\sum}}_{1 \leq k < j \leq r}\Bigg[\frac{g_{35}(s, t)}{h_{G}^{2}} W\Bigg(\frac{T_{ij} - s}{h_G}\Bigg) W\Bigg(\frac{T_{ik} - t}{h_G}\Bigg) G(T_{ij}, T_{ik})\Bigg]  -  \frac{1}{6}(\phi(s, t) + \rho(s, t) + \tau(s, t))\\ \nonumber
&=&\frac{1}{n} \mathlarger{\mathlarger{\sum}}_{i \leq n} \frac{2}{r(r-1)}\mathlarger{\mathlarger{\sum}}_{1 \leq k < j \leq r}\frac{g_{35}(s, t)}{h_{G}^{2}} W\Bigg(\frac{T_{ij} - s}{h_G}\Bigg) W\Bigg(\frac{T_{ik} - t}{h_G}\Bigg) \Bigg[G(T_{ij}, T_{ik}) -  \frac{1}{6}\Bigg(\frac{p_{1} + p_{2} + p_{3}}{A_{11}}\Bigg) G(s, t)\\ \nonumber 
&\qquad - &\Bigg(\frac{p_{4} + p_{5} + p_{6}}{A_{11}}\Bigg) (T_{ij} - s) \partial_{s}{G} (s,t) - \Bigg(\frac{p_{7} + p_{8} + p_{9}}{A_{11}}\Bigg)(T_{ik} - t) \partial_{t}{G} (s,t)\Bigg] \\
&=& \bigo (h_{G}^{d + 1}), \quad \mbox{a.s uniformly on}~ 0 \leq t \leq s \leq 1.
\end{eqnarray}}}\hfill$\Box$
\begin{singlespace}
  \noindent {\bf Proof of $(ii)_2$\,.}
Moving on to the proof of second part of Theorem \ref{Thm-Main-1} (ii), we have,
\begin{eqnarray*}
h_{G} \widehat{\partial_{s}G}(s, t) - h_{G} \partial_{s} G(s, t) 
&=& \Bigg(\dfrac
{g_{27}(s,t) + g_{36}(s,t) + g_{37}(s,t) - g_{38}(s,t) - g_{39}(s,t) - g_{40}(s,t)}{g_{7}(s,t) + g_{8}(s,t) + g_{9}(s,t) 
-2g_{10}(s,t) - 1}\Bigg)\\
&\qquad -& (\phi_{1}(s, t) + \rho_{1}(s, t) + \tau_{1}(s, t)),
\end{eqnarray*}
where,
\begin{eqnarray*}
g_{27}(s, t) &=& \dfrac{\sum \limits_{i \leq n} \sum \limits_{k < j}(T_{ik} - t)K(.,.) \times \sum \limits_{i \leq n} \sum \limits_{k < j} Y_{ik}Y_{ij}(T_{ij} - s) K(.,.)\times \sum \limits_{i \leq n} \sum \limits_{k < j}(T_{ik} - t) K(.,.) }{\sum \limits_{i \leq n} \sum \limits_{k < j} K(.,.)\times \sum \limits_{i \leq n} \sum \limits_{k < j}(T_{ij} - s)^{2} K(.,.) \times \sum \limits_{i \leq n} \sum \limits_{k < j}(T_{ik} - t)^{2} K(.,.)},\\
g_{36}(s, t) &=& \dfrac{\sum \limits_{i \leq n} \sum \limits_{k < j} Y_{ik} Y_{ij}K(.,.) \times \sum \limits_{i \leq n} \sum \limits_{k < j} (T_{ij} - s) K(.,.)}{\sum \limits_{i \leq n} \sum \limits_{k < j} K(.,.)\times \sum \limits_{i \leq n} \sum \limits_{k < j}(T_{ij} - s)^{2} K(.,.)},\\
g_{37}(s, t) &=& \dfrac{\sum \limits_{i \leq n} \sum \limits_{k < j} (T_{ij} - s)(T_{ik} - t) K(.,.) \times \sum \limits_{i \leq n} \sum \limits_{k < j} Y_{ik} Y_{ij} (T_{ik} - t) K(.,.)}{\sum \limits_{i \leq n} \sum \limits_{k < j} (T_{ij} - s)^2 K(.,.)\times \sum \limits_{i \leq n} \sum \limits_{k < j}(T_{ik} - t)^{2} K(.,.)},\\
g_{38}(s, t) &=& \dfrac{\sum \limits_{i \leq n} \sum \limits_{k < j} Y_{ik}Y_{ij}(T_{ij} - s) K(.,.)}{\sum \limits_{i \leq n} \sum \limits_{k < j} (T_{ij} - s)^2 K(.,.)},\\
g_{39}(s, t) &=& \dfrac{\sum \limits_{i \leq n} \sum \limits_{k < j} Y_{ik} Y_{ij}K(.,.) \times \sum \limits_{i \leq n} \sum \limits_{k < j} (T_{ij} - s)(T_{ik} - t) K(.,.) \times \sum \limits_{i \leq n} \sum \limits_{k < j}(T_{ik} - t) K(.,.) }{\sum \limits_{i \leq n} \sum \limits_{k < j} K(.,.)\times \sum \limits_{i \leq n} \sum \limits_{k < j}(T_{ij} - s)^{2} K(.,.) \times \sum \limits_{i \leq n} \sum \limits_{k < j}(T_{ik} - t)^{2} K(.,.)},\\
g_{40}(s, t) &=& \dfrac{\sum \limits_{i \leq n} \sum \limits_{k < j} (T_{ik} - t) K(.,.) \times \sum \limits_{i \leq n} \sum \limits_{k < j} (T_{ij} - s) K(.,.) \times \sum \limits_{i \leq n} \sum \limits_{k < j} Y_{ik} Y_{ij} (T_{ik} - t) K(.,.) }{\sum \limits_{i \leq n} \sum \limits_{k < j} K(.,.)\times \sum \limits_{i \leq n} \sum \limits_{k < j}(T_{ij} - s)^{2} K(.,.) \times \sum \limits_{i \leq n} \sum \limits_{k < j}(T_{ik} - t)^{2} K(.,.)},\\
K(.,.) &=& K_{{H}_{G}}((T_{ij} - s),(T_{ik} - t)),\\
\phi_{1}(s, t) &=& \Big(\dfrac{p_{10} + p_{11}+ p_{12}}{A_{11}}\Big) G(s, t), \\
\rho_{1}(s,t) &=&  \Big(\dfrac{p_{13} + p_{14} + p_{15}}{A_{11}}\Big) h_{G} \partial_{s}{G} (s,t),\\
\tau_{1}(s,t) &=& \Big(\dfrac{p_{16} + p_{17} + p_{18}}{A_{11}}\Big) h_{G} \partial_{t}{G} (s,t),   
\end{eqnarray*}
{\footnotesize{\begin{eqnarray}\nonumber
p_{10} &=& \big[(adw_{1} + bew_{2} + cfw_{3})(dw_{1} + ew_{2} + fw_{3}) - (aw_{1} + bw_{2} + cw_{3})(d^{2}w_{1} + e^{2}w_{2} + f^{2}w_{3})\big](w_{1} + w_{2} + w_{3}),\\ \nonumber
p_{11} &=& \big[(d^{2}w_{1} + e^{2}w_{2} + f^{2}w_{3})(w_{1} + w_{2} + w_{3}) - (dw_{1} + ew_{2} + fw_{3})^{2}\big](aw_{1} + bw_{2} + cw_{3}),\\ \nonumber
p_{12} &=& \big[(aw_{1} + bw_{2} + cw_{3})(dw_{1} + ew_{2} + fw_{3}) - (adw_{1} + bew_{2} + cfw_{3})(w_{1} + w_{2} + w_{3})\big](dw_{1} + ew_{2} + fw_{3}),\\ \nonumber
p_{13} &=& \big[(adw_{1} + bew_{2} + cfw_{3})(dw_{1} + ew_{2} + fw_{3}) - (aw_{1} + bw_{2} + cw_{3})(d^{2}w_{1} + e^{2}w_{2} + f^{2}w_{3})\big](aw_{1} + bw_{2} + cw_{3}),\\ \nonumber
p_{14} &=& \big[(d^{2}w_{1} + e^{2}w_{2} + f^{2}w_{3})(w_{1} + w_{2} + w_{3}) - (dw_{1} + ew_{2} + fw_{3})^{2}\big](a^{2}w_{1} + b^{2}w_{2} + c^{2}w_{3}),\\ \nonumber
p_{15} &=& \big[(aw_{1} + bw_{2} + cw_{3})(dw_{1} + ew_{2} + fw_{3}) - (adw_{1} + bew_{2} + cfw_{3})(w_{1} + w_{2} + w_{3})\big](adw_{1} + bew_{2} + cfw_{3}),\\ \nonumber
p_{16} &=&  \big[(adw_{1} + bew_{2} + cfw_{3})(dw_{1} + ew_{2} + fw_{3}) - (aw_{1} + bw_{2} + cw_{3})(d^{2}w_{1} + e^{2}w_{2} + f^{2}w_{3})\big](dw_{1} + ew_{2} + fw_{3}),\\ \nonumber
p_{17} &=& \big[(d^{2}w_{1} + e^{2}w_{2} + f^{2}w_{3})(w_{1} + w_{2} + w_{3}) - (dw_{1} + ew_{2} + fw_{3})^{2}\big](adw_{1} + bew_{2} + cfw_{3}),\\
p_{18} &=& \big[(aw_{1} + bw_{2} + cw_{3})(dw_{1} + ew_{2} + fw_{3}) - (adw_{1} + bew_{2} + cfw_{3})(w_{1} + w_{2} + w_{3})\big](d^2 w_{1} + e^2w_{2} + f_{2}w_{3}).
\label{1.2}
\end{eqnarray}}}

\noindent The terms $g_{7}(s,t), g_{8}(s,t), g_{9}(s,t)$ and $g_{10}(s,t)$ are the same as defined in the previous proof. We considered the term $g_{27}(s,t)$ only and analysed the sub term of $g_{27}(s,t)$, which contains $Y_{ik}$ and $Y_{ij}$. All other terms in the numerator and denominator of $h_{G} \widehat{\partial_{s}G}(s, t)$, except the sub term of $g_{27}(s,t)$, which contains $Y_{ik}$ and $Y_{ij}$, can be expressed as a function of $s$ and $t$. The ratio of, all other sub terms of $g_{35}(s,t)$ (which don't contain $Y_{ik}$ and $Y_{ij}$), and the terms in the denominator, is taken as $g_{36}(s, t).$ Similarly, we subtracted $\frac{1}{6}h_{G} \partial_{s} G(s, t)$ from $h_{G} \widehat{\partial_{s}G}(s, t).$ We continued our analysis from here the same way, as is done in the previous proof. We found that the rates of convergence are similar. 
\end{singlespace} \hfill$\Box$

\begin{singlespace}
\noindent {\bf Proof of $(i)_1$\,.} 
 For $d = 2$, we partially differentiate the sum of squared errors in Equation \ref{beta}, with respect to the parameters, and set these equations to zero respectively. The corresponding equations are:
 
\begin{equation}  
    \widehat{m}(t) = \frac{f_{41}(t) + f_{42}(t) + f_{43}(t) - f_{44}(t) - f_{45}(t) - f_{46}(t)}{f_{47}(t) + f_{48}(t) + f_{49}(t) - 2f_{50}(t) - 1},
 \end{equation} 
where,\\
$f_{41}(t) = \dfrac{g_{411}(t) \times (g_{412}(t))^{2}}{\prod \limits_{i=4}^{6} g_{41i}(t)},~~~
f_{42}(t) = \dfrac{\prod \limits_{i=1}^{2} g_{42i}(t)}{\prod \limits_{i=4}^{5} g_{41i}(t)},~~~
f_{43}(t) = \dfrac{g_{415}(t) \times g_{432}(t)}{g_{414}(t) \times g_{416}(t)},$\\
$f_{44}(t) = \dfrac{\sum \limits_{i \leq n} \sum \limits_{j \leq r(n)} Y_{ij} K_{h_{m}^{2}}(T_{ij} - t)}{\sum \limits_{i \leq n} \sum \limits_{j \leq r(n)} K_{h_{m}^{2}}(T_{ij} - t)},~~
f_{45}(t) = \dfrac{g_{421}(t) \times g_{412}(t) \times g_{432}(t)}{\prod \limits_{i=4}^{6} g_{41i}(t)},$\\
$f_{46}(t) = \dfrac{g_{415}(t) \times g_{422}(t) \times g_{412}(t)}{\prod \limits_{i=4}^{6} g_{41i}(t)},~~
f_{47}(t) = \dfrac{(g_{421}(t))^{2}}{\prod \limits_{i=4}^{5} g_{41i}(t)},~~
f_{48}(t) = \dfrac{(g_{415}(t))^{2}}{g_{414}(t) \times g_{416}(t)},$\\
$f_{49}(t) = \dfrac{(g_{412}(t))^{2}}{\prod \limits_{i=5}^{6} g_{41i}(t)}, ~\mbox{and}~ 
f_{50}(t) =  \dfrac{g_{421}(t) \times g_{412}(t) \times g_{415}(t)}{\prod \limits_{i=4}^{6} g_{41i}(t)}$.\\
\noindent The $g_{()}'s$ in the aforementioned functions are:
 \begin{align*}
g_{411}(t) &= \sum \limits_{i \leq n} \sum \limits_{j \leq r(n)} Y_{ij} K_{h_{m}^{2}}(T_{ij} - t),\\
g_{412}(t) &= \sum \limits_{i \leq n} \sum \limits_{j \leq r(n)} (T_{ij} - t)^{3} K_{h_{m}^{2}}(T_{ij} - t),\\
g_{414}(t) &= \sum \limits_{i \leq n} \sum \limits_{j \leq r(n)} K_{h_{m}^{2}}(T_{ij} - t),\\
g_{415}(t) &= \sum \limits_{i \leq n} \sum \limits_{j \leq r(n)} (T_{ij} - t)^{2} K_{h_{m}^{2}}(T_{ij} - t),\\
g_{416}(t) &= \sum \limits_{i \leq n} \sum \limits_{j \leq r(n)} (T_{ij} - t)^{4} K_{h_{m}^{2}}(T_{ij} - t),\\
g_{421}(t) &= \sum \limits_{i \leq n} \sum \limits_{j \leq r(n)} (T_{ij} - t) K_{h_{m}^{2}}(T_{ij} - t),\\
g_{422}(t) &= \sum \limits_{i \leq n} \sum \limits_{j \leq r(n)} Y_{ij} (T_{ij} - t) K_{h_{m}^{2}}(T_{ij} - t),\\ 
g_{432}(t) &= \sum \limits_{i \leq n} \sum \limits_{j \leq r(n)} Y_{ij} (T_{ij} - t)^{2} K_{h_{m}^{2}}(T_{ij} - t).
\end{align*}

\noindent Similarly, we obtained,
\begin{equation}
h_{m}
 \widehat{\partial m}(t) = \frac{g_{41}(t) + g_{42}(t) + g_{43}(t) - g_{44}(t) - g_{45}(t) - g_{46}(t)}{g_{47}(t) + g_{48}(t) + g_{49}(t) - 2g_{50}(t) - 1},
\end{equation}
where,\\
$g_{41}(t) = \dfrac{(f_{411}(t))^{2} \times f_{412}(t)}{\prod \limits_{i=4}^{6} f_{41i}(t)},~~
g_{42}(t) = \dfrac{\prod \limits_{i=1}^{2} f_{42i}(t)}{\prod \limits_{i=4}^{5} f_{41i}(t)},~~
g_{43}(t) = \dfrac{\prod \limits_{i=1}^{2} f_{43i}(t)}{\prod \limits_{i=5}^{6} f_{41i}(t)},$\\
$g_{44}(t) = \dfrac{\sum \limits_{i \leq n} \sum \limits_{j \leq r(n)} Y_{ij} (T_{ij} - t) K_{h_{m}^{2}}(T_{ij} - t)}{\sum \limits_{i \leq n} \sum \limits_{j \leq r(n)} (T_{ij} - t)^{2} K_{h_{m}^{2}}(T_{ij} - t)},~~
g_{45}(t) = \dfrac{f_{421}(t) \times f_{432}(t) \times f_{415}(t)}{\prod \limits_{i=4}^{6} f_{41i}(t)},$\\
$g_{46}(t) = \dfrac{f_{415}(t) \times f_{422}(t) \times f_{431}(t)}{\prod \limits_{i=4}^{6} f_{41i}(t)},~~
g_{47}(t) = \dfrac{(f_{422}(t))^{2}}{\prod \limits_{i=4}^{5} f_{41i}(s,t)},~~
g_{48}(t) = \dfrac{(f_{411}(t))^{2}}{f_{414}(t) \times f_{416}(t)},$\\
$g_{49}(t) = \dfrac{(f_{432}(t))^{2}}{\prod \limits_{i=5}^{6} f_{41i}(t)} ~\mbox{and}~ 
g_{50}(t) =  \dfrac{f_{422}(t) \times f_{432}(t) \times f_{411}(t)}{\prod \limits_{i=4}^{6} f_{41i}(t)}.$\\

\noindent The $f_{()}'s$ in the aforementioned functions are:
\begin{align*}
f_{411}(t) &= \sum \limits_{i \leq n} \sum \limits_{j \leq r(n)} (T_{ij} - t)^{2} K_{h_{m}^{2}}(T_{ij} - t),\\
f_{412}(t) &= \sum \limits_{i \leq n} \sum \limits_{j \leq r(n)} Y_{ij} (T_{ij} - t) K_{h_{m}^{2}}(T_{ij} - t),\\
f_{414}(t) &= \sum \limits_{i \leq n} \sum \limits_{j \leq r(n)} K_{h_{m}^{2}}(T_{ij} - t),\\
f_{415}(t) &= \sum \limits_{i \leq n} \sum \limits_{j \leq r(n)} (T_{ij} - t)^{2} K_{h_{m}^{2}}(T_{ij} - t),\\
f_{416}(t) &= \sum \limits_{i \leq n} \sum \limits_{j \leq r(n)} (T_{ij} - t)^{4} K_{h_{m}^{2}}(T_{ij} - t),\\
f_{421}(t) &= \sum \limits_{i \leq n} \sum \limits_{j \leq r(n)} Y_{ij} K_{h_{m}^{2}}(T_{ij} - t),\\
f_{422}(t) &= \sum \limits_{i \leq n} \sum \limits_{j \leq r(n)} (T_{ij} - t) K_{h_{m}^{2}}(T_{ij} - t),\\
f_{431}(t) &= \sum \limits_{i \leq n} \sum \limits_{j \leq r(n)} Y_{ij} (T_{ij} - t)^{2} K_{h_{m}^{2}}(T_{ij} - t)~\mbox{and}~\\
f_{432}(t) &= \sum \limits_{i \leq n} \sum \limits_{j \leq r(n)} (T_{ij} - t)^{3} K_{h_{m}^{2}}(T_{ij} - t).
\end{align*}

\noindent Now, we should obtain $m(t)$. Following Mohammadi and Panaretos \cite{B3}, we observe that  
\begin{eqnarray*}
\begin{bmatrix}
	m(t)\\
	h_{m} \partial{m}(t)\\
	\beta_{2}
\end{bmatrix}
= \big[X_{1}^T W_{m} X_{1}\big]^{-1} \big[X_{1}^T W_{m} X_{1}\big] 
\begin{bmatrix}
	m(t)\\
	h_{m} \partial{m}(t)\\
	\beta_{2}
\end{bmatrix},
\end{eqnarray*}\\
where,
\vspace{0.1cm}
\begin{align}\label{def-X-matrix1}
    X_{1} := \begin{bmatrix}
1 & a & d_{1}\\
1 & b & e_{1}\\
1 & c & f_{1}
\end{bmatrix},
\end{align}
the elements $a, b, c, d_{1}, e_{1}, f_{1}$ in $X_{1}$ are defined as
\begin{eqnarray*}
 a =  T_{1j} - t,  \quad b =  T_{2j} - t, \quad c =  T_{3j} - t, \\ 
 d_{1} =  (T_{1j} - t)^{2}, \quad e_{1} =  (T_{2j} - t)^{2}, \quad f_{1} =  (T_{3j} - t)^{2},
\end{eqnarray*}
and the matrix $W_{m}$ is defined as
$$
W_{m} := \begin{bmatrix}
    w_{1m} & 0 & 0\\
    0 & w_{2m} & 0\\
    0 & 0 & w_{3m}
\end{bmatrix},
$$
where, for $i=1, 2$ and $3$, $w_{im}$ is defined as
\[
w_{im} :=  \frac{1}{h_{m}^{2}} W\Bigg(\frac{T_{ij} - t}{h_m^{2}}\Bigg) = K_{h_{m}^{2}}(T_{ij} - t)\,.
\]

\noindent Solving the aforementioned matrix setup, we obtain,
\begin{eqnarray}
m(t) = \phi_{2} (t) + \rho_{2} (t) + \tau_{2}(t). \label{e12}
\end{eqnarray}
In Equation (\ref{e12}),\\ 
$\phi_{2}(t) = \Big(\dfrac{p_{1m} + p_{2m}+ p_{3m}}{A_{12}}\Big) m(t),~ \rho_{2}(t) =  \Big(\dfrac{p_{4m} + p_{5m} + p_{6m}}{A_{12}}\Big) h_{m} \partial{m}(t), \tau_{2}(t) = \Big(\dfrac{p_{7m} + p_{8m} + p_{9m}}{A_{12}}\Big) \beta_{2},$\\

\noindent where, $p_{1m}, p_{2m}, p_{3m}, p_{4m}, p_{5m}, p_{6m}, p_{7m}, p_{8m}, p_{9m},$ have the same representation (except $d, e, f, w_{1}, w_{2}, w_{3}$ are replaced by $d_{1}, e_{1}, f_{1}, w_{1m}, w_{2m}~ \mbox{and}~ w_{3m}$, respectively) as defined in Equation \ref{1.1}, in the proof of Theorem - $5.2 (ii)_{1},$ and $A_{12}$ represents the determinant of the matrix $(X_{1}^{T} W_{m}X_{1}).$ 

\noindent Our objective is to find $\widehat{m}(t)-m(t)$. Hence, we obtain 
\onehalfspacing
\begin{eqnarray}\nonumber
\underset{0 \leq t \leq 1}{\sup}|\widehat{m}(t)-m(t)|
&=&  \underset{0 \leq t \leq 1}{\sup} \Bigg|\frac{f_{41}(t) + f_{42}(t) + f_{43}(t) - f_{44}(t) - f_{45}(t) - f_{46}(t)}{f_{47}(t) + f_{48}(t) + f_{49}(t) - 2f_{50}(t) - 1} \\ \nonumber
 &\qquad -& (\phi_{2}(t) + \rho_{2}(t) + \tau_{2}(t))\Bigg|.
\end{eqnarray} 
\onehalfspacing
Now, we are considering only the first term of the numerator in $\widehat{m}(t)$, i.e., $f_{41}(t)$, because the expression is very long. However, the denominator remains the same. Therefore,
\begin{eqnarray}\nonumber
\underset{0 \leq t \leq 1}{\sup}|\widehat{m}(t)-m(t)|
&=&  \underset{0 \leq t \leq 1}{\sup} \Bigg|\frac{f_{41}(t)}{f_{47}(t) + f_{48}(t) + f_{49}(t) - 2f_{50}(t) - 1}\\ \nonumber
& \quad \quad -& (\phi_{2}(t) + \rho_{2} (t) + \tau_{2}(t))\Bigg|\\ \nonumber
&=& \underset{0 \leq t \leq 1}{\sup} \Bigg|\frac{g_{411}(t) \times (g_{412}(t))^{2}}{(\prod \limits_{i = 4}^{6} g_{41i}(t)(f_{47}(t) + f_{48}(t) + f_{49}(t) - 2f_{50}(t) - 1)} \\ 
&\qquad -& (\phi_{2}(t) + \rho_{2} (t) + \tau_{2}(t))\Bigg|.
\label{e21}
\end{eqnarray}
According to Hsing and Eubank \cite{B4}, the terms, $g_{412}(t), g_{414}(t), g_{415}(t), g_{416}(t), f_{47}(t), f_{48}(t), f_{49}(t)$ and $f_{50}(t)$, which have summations, can again be expressed as some function of $t$. However, we are not considering the term $g_{411}(t)$ in such context, because it contains the term $Y_{ij}$, which is really needed for further analysis.  Hence, Equation (\ref{e21}) can be expressed as 
\begin{eqnarray}\nonumber
\underset{0 \leq t \leq 1}{\sup}|\widehat{m}(t)-m(t)|
&=& \underset{0 \leq t \leq 1}{\sup} \Bigg|\frac{g_{411}(t) \times g_{417}(t)}{(g_{418}(t) \times g_{419}(t) \times g_{420}( t)}\\\nonumber
& \qquad \times& \frac{1}{(g_{423}(t) + g_{424}(t) + g_{425}(t) - 2g_{426}(t) - 1)} \\ \nonumber 
& \qquad -& (\phi_{2} (t) + \rho_{2} (t) + \tau_{2}(t))\Bigg|.\\\nonumber
\end{eqnarray}
Therefore, 
$$\underset{0 \leq t \leq 1}{\sup}|\widehat{m}(t)-m(t)| = \underset{0 \leq t \leq 1}{\sup} \Bigg|(g_{411}(t) \times g_{435}(t)) - (\phi_{2} (t) + \rho_{2} (t) + \tau_{2}(t))\Bigg|,$$
where, $$g_{435}(t) = \dfrac{g_{417}(t)}{(g_{418}(t) \times g_{419}(t) \times g_{420}(t))(g_{423}(t) + g_{424}(t) + g_{425}(t) - 2g_{426}(t) - 1)}.$$

\noindent Proceeding further, Equation (\ref{e21}) can be written as 
\begin{eqnarray} \nonumber
\underset{0 \leq t \leq 1}{\sup}|\widehat{m}(t)-m(t)|
 &=& \underset{0 \leq t \leq 1}{\sup}\Bigg|\frac{1}{n} \mathlarger{\mathlarger{\sum}}_{i \leq n}  \frac{2}{r(r-1)}\mathlarger{\mathlarger{\sum}}_{j \leq r(n)} \bigg[Y_{ij} K_{h_{m}^{2}}(T_{ij} - t) \times  g_{435}(t)\bigg]\\  \nonumber 
 &\qquad -& (\phi_{2}(t) + \rho_{2}(t) + \tau_{2}(t))\Bigg|.
\end{eqnarray}
 From here onwards, we continue our analysis in a similar sequence, as done in the proof of Theorem - $5.2 (ii)_{1}.$ Since, lines of proof are similar, we have not presented the whole proof. However, one major change is that we have used 
$\mathcal{B}_{h_m}(n),$ 
instead of $\mathcal{Q}_{h_G}(n)$ in the derivation, as defined in Theorem \ref{Thm-Main-1}.
After derivation, we found that the rates of convergence are $\bigo(h_{m}^{d + 1}).$\hfill$\Box$

\noindent {\bf Proof of $(i)_2$\,.}
 For proof of second part of Theorem \ref{Thm-Main-1} (i), we have,
\begin{eqnarray*}
h_{m} \widehat{\partial m}(t) - h_{m} \partial m(t) 
&=& \Bigg(\dfrac{g_{51}(t) + g_{52}(t) + g_{53}(t) - g_{54}(t) - g_{55}(t) - g_{56}(t)}{g_{47}(t) + g_{48}(t) + g_{49}(t) - 2g_{50}(t) - 1}\\
&\qquad -& (\phi_{3}(t) + \rho_{3}(t) + \tau_{3}(t))\Bigg),
\end{eqnarray*}
where,
\begin{align*}
g_{51}(t) &= \dfrac{\sum \limits_{i \leq n} \sum \limits_{j \leq r(n)}(T_{ij} - t)^{2} K(.) \sum \limits_{i \leq n} \sum \limits_{j \leq r(n)} Y_{ij}(T_{ij} - t) K(.) \times \sum \limits_{i \leq n} \sum \limits_{j \leq r(n)}(T_{ij} - t)^{2} K(.) }{\sum \limits_{i \leq n} \sum \limits_{j \leq r(n)} K(.)\times \sum \limits_{i \leq n} \sum \limits_{j \leq r(n)}(T_{ij} - t)^{2} K(.) \times \sum \limits_{i \leq n} \sum \limits_{j \leq r(n)}(T_{ij} - t)^{4} K(.)},\\
g_{52}(t) &= \dfrac{\sum \limits_{i \leq n} \sum \limits_{j \leq r(n)} Y_{ij}K(.) \times \sum \limits_{i \leq n} \sum \limits_{j \leq r(n)} (T_{ij} - t) K(.)}{\sum \limits_{i \leq n} \sum \limits_{j \leq r(n)} K(.)\times \sum \limits_{i \leq n} \sum \limits_{j \leq r(n)}(T_{ij} - t)^{2} K(.)},\\
g_{53}(t) &= \dfrac{\sum \limits_{i \leq n} \sum \limits_{j \leq r(n)} Y_{ij} (T_{ij} - t)^{2} K(.) \times \sum \limits_{i \leq n} \sum \limits_{j \leq r(n)} (T_{ij} - t)^{3} K(.)}{\sum \limits_{i \leq n} \sum \limits_{j \leq r(n)} (T_{ij} - t)^{4} K(.)\times \sum \limits_{i \leq n} \sum \limits_{j \leq r(n)} (T_{ij} - t)^{2} K(.)},\\
g_{54}(t) &= \dfrac{\sum \limits_{i \leq n} \sum \limits_{j \leq r(n)} Y_{ij}(T_{ij} - t) K(.)}{\sum \limits_{i \leq n} \sum \limits_{j \leq r(n)} (T_{ij} - t)^2 K(.)},\\
g_{55}(t) &= \dfrac{\sum \limits_{i \leq n} \sum \limits_{j \leq r(n)} Y_{ij} K(.) \times \sum \limits_{i \leq n} \sum \limits_{j \leq r(n)} (T_{ij} - t)^{3} K(.) \times \sum \limits_{i \leq n} \sum \limits_{j \leq r(n)} (T_{ij} - t)^{2} K(.)}{\sum \limits_{i \leq n} \sum \limits_{j \leq r(n)} K(.) \times \sum \limits_{i \leq n} \sum \limits_{j \leq r(n)} (T_{ij} - t)^{2} K(.) \times \sum \limits_{i \leq n} \sum \limits_{j \leq r(n)} (T_{ij} - t)^{4} K(.)},\\
g_{56}(t) &= \dfrac{\sum \limits_{i \leq n} \sum \limits_{j \leq r(n)} (T_{ij} - t)^{2} K(.) \times \sum \limits_{i \leq n} \sum \limits_{j \leq r(n)} (T_{ij} - t) K(.) \times \sum \limits_{i \leq n} \sum \limits_{j \leq r(n)} Y_{ij} (T_{ij} - t)^{2} K(.)}{\sum \limits_{i \leq n} \sum \limits_{j \leq r(n)} K(.) \times \sum \limits_{i \leq n} \sum \limits_{j \leq r(n)} (T_{ij} - t)^{2} K(.) \times \sum \limits_{i \leq n} \sum \limits_{j \leq r(n)} (T_{ij} - t)^{4} K(.)},\\
K(.) &= K_{h_{m}^{2}}(T_{ij} - t),\\
 \phi_{3}(t) &= \Big(\dfrac{p_{10m} + p_{11m}+ p_{12m}}{A_{12}}\Big) m(t), \\
\rho_{3}(t) &=  \Big(\dfrac{p_{13m} + p_{14m} + p_{15m}}{A_{12}}\Big) h_{m} \partial {m} (t),\\
\tau_{3}(t) &= \Big(\dfrac{p_{16m} + p_{17m} + p_{18m}}{A_{12}}\Big) \beta_{2}, 
\end{align*}
$p_{10m}, p_{11m}, p_{12m}, p_{13m}, p_{14m}, p_{15m}, p_{16m}, p_{17m}, p_{18m},$ have the same representation (except $d, e, f, w_{1}, w_{2}, w_{3}$ are replaced by $d_{1}, e_{1}, f_{1}, w_{1m}, w_{2m}~ \mbox{and}~ w_{3m}$, respectively) as defined in equation \ref{1.2}, in the proof of Theorem $5.2 (ii)_{2}.$ The terms $g_{47}(t), g_{48}(t), g_{49}(t)$ and $g_{50}(t)$ are the same as defined in the previous proof.  We considered the term $g_{51}(t)$ only and analysed the sub term of $g_{51}(t)$, which contains $Y_{ij}$. All other terms in the numerator and denominator of $h_{m} \widehat{\partial m}(t)$, except the sub term of $g_{51}(t)$, which contains $Y_{ij}$, can be expressed as a function of $t$. The ratio of, all other sub terms of $g_{51}(t)$ (which don't contain $Y_{ij}$), and the terms in the denominator, is taken as $g_{57}( t).$ Similarly, we subtracted $\frac{1}{6}h_{m} \partial m(t)$ from $h_{m} \widehat{\partial m}(t).$ We continued our analysis from here the same way, as is done in the previous proof. We found that the rates of convergence are similar.\hfill$\Box$    
\end{singlespace}

\vspace{0.5cm}
\noindent {\bf Acknowledgment: } Subhra Sankar Dhar gratefully acknowledges his core research grant (CRG/2022/001489), Government of India.

\vspace{0.3cm}

\newpage
{\bf Supplement to ``Least Square Estimation: SDE Perturbed by L\'evy Noise with Sparse Sample Paths"}

{\bf Abstract :} In this supplementary materials, the proposed methodology is applied to a benchmark dataset of functional data/curves, and a small simulation study is conducted to illustrate the findings.






 	
 	


\section{Simulation Studies}\label{SS}
In this section, we have conducted some simulation studies to see the performance of the proposed methodologies/estimators in the case of finite samples. Motivated by the model considered in \cite{Zhang}, here we consider the following model : 

\begin{align}\label{EX-1}
dX(t)= -(2 + \sin t)X(t)\,dt + \frac{1}{2}\sin t\, dB(t)+ \int_{|y|\leq 1} \sin t\,\tilde{\eta}(dt,dy),   
\end{align}  where $t\geq 0$. Note that in comparison with model described in (1.1) in the main manuscript, we have $\mu (t) = - (2 + \sin t)$, $\sigma (t) = \frac{1}{2}\sin t$ and $\xi (t) = \sin t$. 

\vspace{0.25in}

\noindent {\bf Case 1:} Suppose that $B(t)$ is a random element associated with standard Brownian motion, and for numerical studies, we consider $t\in [0, 50]$, and n many curves are generated independently using the model described in \eqref{EX-1}.
Let $\widehat{\mu}_{{{\tt LSE}}}(t)$ (essentially, same as $\widehat{\mu}(t)$ defined in Equation (4.6) be the {\tt LSE} estimator of $\mu(t)$ based on the aforesaid $n$ many curves. Next, we compute the empirical mean squared error ({\tt EMSE}) of $\widehat{\mu}_{{{\tt LSE}}}(t)$, which is defined as $${\tt EMSE}(\widehat{\mu}_{{{\tt LSE}}}(t)) : = \int\limits_{0}^{50}\{\widehat{\mu}_{{{\tt LSE}}}(t) - \mu(t)\}^{2} dt.$$ In Figure \ref{fig:1}, we have plotted the values of ${\tt EMSE}(\widehat{\mu}_{{{\tt LSE}}}(t))$ for different choices of $n$. It is indicated from this diagram that ${\tt EMSE}(\widehat{\mu}_{{{\tt LSE}}}(t))$ decreases as $n$ increases, which further supports the theoretical finding (see Theorem 5.3) that $\widehat{\mu}_{{{\tt LSE}}}(t)$ converges uniformly to $\mu(t)$ with probability 1.

\begin{figure}[htbp!]
 \centering
     \includegraphics[width=4.5in, height=4.5 in]{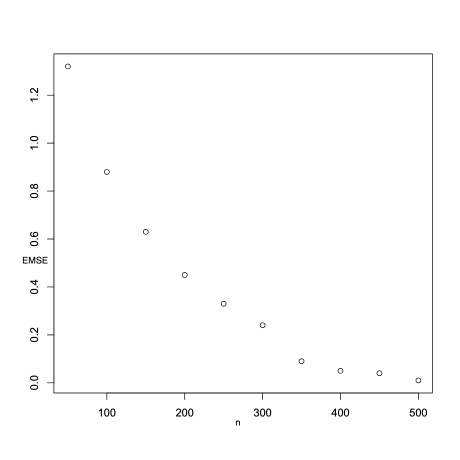}
   \caption{Plot of EMSE for various choices of $n$}
    \label{fig:1}
\end{figure}

\begin{figure}[htbp!]
 \centering
     \includegraphics[width=4.5in, height=4.5 in]{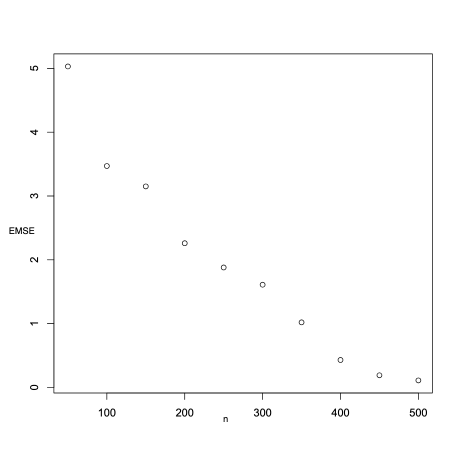}
   \caption{Plot of EMSE for various choices of $n$}
    \label{fig:2}
\end{figure}


\vspace{0.25in}

\noindent {\bf Case 2:} Suppose that $B(t)$ is a random element associated with zero mean Gaussian process with covariance kernel $k(s, t) = 50\min(s, t)$ for all $s\in [0, 50]$ and $t\in [0, 50]$. In this case, in Figure \ref{fig:2}, we have plotted the values of EMSE for different choices of $n$. It is indicated by this diagram that here EMSE values are larger relative to the EMSE values for Case 1 (see Figure \ref{fig:1}), although here also the EMSE value decreases as $n$ increases. Hence, in this case also, the study further favour the theoretical result that $\widehat{\mu}_{{{\tt LSE}}}(t)$ converges uniformly to $\mu(t)$ with probability 1. Overall, the LSE estimator of $\mu(.)$ performs better in Case 1 compared to Case 2 because in this case, the data are generated from the process having marginal distributions with heavy tail.

In the end, here also, we should mention that one can do a similar study for $\sigma(t)$ and $\xi(t)$.




\vspace{0.3cm}

\section{Real Data Analysis}\label{RDA}
In this section, well-known Canadian Weather data is analysed. This data set consists of daily temperature and precipitation at $n = 35$ locations in Canada averaged over 1960 to 1994. \cite{Dette2024} studied this data set in the context of functional regression, and \cite{Bhar2024} analysed this data set regarding the test for independence of infinite diemnsional random elements. A version of the data is available in \url{https://climate.weather.gc.ca/historical_data/search_historic_data_e.html}. The dimension of the data equals 365, which is much larger than the sample size $= 35$ of the data, and therefore embedding such a high-dimensional data into an infinite-dimensional space (specifically speaking, functional in nature) is legitimate enough, and hence the measure associated with the stochastic process involved in model described in (1.1) in the main manuscript can be considered as parent measure for generating these data. In addition, to make it the time parameter space $[0, 1]$, the temperature and the precipitation of 365
days are considered as the equally spaced 365 time points over $[0, 1]$. 


Now, let this data set be tied to the model described in (1.1) in the main manuscript and compute the least squares estimator of $\mu(t)$, $\sigma_{D}^{2}(t)$ and $\xi(t)$ using (4.6), (4.7) and (4.8), respectively in the main manuscript, at $t = \frac{1}{2}$. We obtain $\widehat{\mu}_{{\tt LSE}}(t)$, $\widehat{\sigma}_{D, {\tt LSE}}^{2}(t)$ and $\widehat{\xi}_{{\tt LSE}}^{2}(t)$ as -2.65, 0.32 and 1.78, respectively. Next, in order to investigate the performance of $\widehat{\mu}_{{\tt LSE}}(t)$, $\widehat{\sigma}_{D, {\tt LSE}}^{2}(t)$ and $\widehat{\xi}_{{\tt LSE}}^{2}(t)$, we compute the Bootstrap Mean Squared Error (BMSE) in the following way. 

Suppose that $B$ many bootstrap resamples are generated with the same size from the original data, and let $\widehat{\mu}_{{\tt i, LSE}}(t)$, $\widehat{\sigma}_{D, {\tt i, LSE}}^{2}(t)$ and $\widehat{\xi}_{{\tt i, LSE}}^{2}(t)$ denote the values of $\widehat{\mu}_{{\tt LSE}}(t)$, $\widehat{\sigma}_{D, {\tt LSE}}^{2}(t)$ and $\widehat{\xi}_{{\tt LSE}}^{2}(t)$ for $i$-th Bootstrap resample, respectively. Here $i = 1, \ldots, B$. Now, 

$$BMSE(\widehat{\mu}_{{\tt LSE}}(t)) = \frac{1}{B}\sum\limits_{i = 1}^{B}(\widehat{\mu}_{{\tt i, LSE}}(t) - \widehat{\mu}_{{\tt LSE}}(t))^{2},$$

$$BMSE(\widehat{\sigma}_{D, {\tt LSE}}^{2}(t)) = \frac{1}{B}\sum\limits_{i = 1}^{B}(\widehat{\sigma}_{D, {\tt i, LSE}}^{2}(t) - \widehat{\sigma}_{D, {\tt LSE}}^{2}(t))^{2},$$ and 

$$BMSE(\widehat{\xi}_{{\tt i, LSE}}^{2}(t)) = \frac{1}{B}\sum\limits_{i = 1}^{B}(\widehat{\xi}_{{\tt i, LSE}}^{2}(t) - \widehat{\xi}_{{\tt LSE}}^{2}(t))^{2}.$$

At $t = \frac{1}{2}$ and $B = 1000$, we obtain $BMSE(\widehat{\mu}_{{\tt LSE}}(t))$, $BMSE(\widehat{\sigma}_{D, {\tt LSE}}^{2}(t))$ and $BMSE(\widehat{\xi}_{{\tt i, LSE}}^{2}(t))$ as 1.32, 0.87 and 1.11, respectively. This study suggests that even for a fairly large number of bootstrap replications, the BMSE values are not small enough relative to zero. Hence, in view of this fact, one may conclude that the LSE does not perform well for this data, and the possible reason may be the presence of outliers/influential observations. This analysis indicates that for parameter estimation in SDE model with jumps, one may consider LAD or in general, quantile estimator to have a more robust version of the estimator, which may be of interest for future research.
\vspace{0.3cm}


\vspace{0.3cm}


\begin{thebibliography}{99}

\bibitem{Applebaum_2004} Applebaum, D. (2004). {\em L\'evy Processes and Stochastic Calculus}, Cambridge Studies in Advanced Mathematics, University Press.

\bibitem{Bhar2024} Bhar, S. and Dhar, S. S. (2024). {\em Testing Independence of Infinite Dimensional Random Elements: A Sup-norm Approach}. \url{https://arxiv.org/abs/2301.00375}, Under revision. 

\bibitem{B5} Boucheron, S., G. Lugosi, \& P. Massart (2013). {\em Concentration Inequalities: A Nonasymptotic Theory
of Independence}. OUP Oxford. 


\bibitem{Bruss_1991} Bruss, F.T., and Rogers, L.C.G. (1991). {\em Pascal processes and their characterization}. Stochastic Processes and their Applications, {\bf 37}, 331--338.


\bibitem{ZB+UM+AAP_2019} Brze\'zniak, Z., Manna, U., Panda, Akash A. (2019). {\em Martingale solutions of nematic liquid crystals driven by pure jump noise in the Marcus canonical form}, J. Differential Equations, {\bf 266}, no. 10, 6204–-6283.

\bibitem{Dette2024} Dette, H. and Tang, J. (2024) {\em Statistical inference for function-on-function linear regression}, Bernoulli, {\bf 30}, 304--331. 



\bibitem{Dhar2014} Dhar, S. S., Chakraborty, B. and Chaudhuri, P. (2014) {\em Comparison of Multivariate Distributions Using Quantile-Quantile Plots and Related Tests}, Bernoulli, {\bf 20}, no. 3, 1484--1506. 

\bibitem{Dhar2023} Dhar, S. S. and Wu, W. (2023) {\em Comparing time varying regression quantiles under shift invariance}, Bernoulli, {\bf 29}, no. 2, 1527--1554.

\bibitem{Sato_2010} Duquesne, T., Reichmann, O., Sato K.-i., Schwab, C. (2010). L\'evy Matters I - A Subseries on L\'evy Processes: {\em Recent Progress in Theory and Applications: Foundations, Trees and Numerical Issues in Finance}.


\bibitem{B2} Fan, J. (2018). {\em Local polynomial modelling and its applications: monographs on statistics and applied probability} 66. Routledge.

\bibitem{B4} Hsing, T., \& Eubank, R. (2015). {\em Theoretical foundations of functional data analysis, with an introduction to linear operators}, Vol. {\bf 997}. John Wiley \& Sons.

\bibitem{Jones_2017} Jones, P. W., \& Smith, P. (2017). {\em Stochastic Processes: An Introduction}, Third Edition, CRC Press, A Chapman \& Hall Book.

\bibitem{Klenke_2008} Klenke, A. (2008), {\em The Poisson Point Process}, Probability Theory: A Comprehensive Course, London: Springer, 525-–542.


\bibitem{B1} Koenker, R., \& Bassett Jr, G. (1978). {\em Regression quantiles}. Econometrica: journal of the Econometric Society, 33-50.

\bibitem{Kou_2002} Kou, S., (2002). {\em A jump-diﬀusion model for option pricing}, Management Science, {\bf 48} , pp. 1086–-1101.

\bibitem{Kunita_2019} Kunita, H., (2019). {\em Stochastic Flows
and Jump-Diffusions}, Probability Theory and Stochastic Modelling
Volume {\bf 92}, Springer.


\bibitem{Mazzola_2011} Mazzola, E. (2011). {\em Reviewing alternative characterizations
of Meixner process}, Probability Surveys, {\bf Vol. 8}, 127--154.


\bibitem{Manna+Panda_2021} Manna, U., \& Panda, Akash A. (2021). {\em Well-posedness and large deviations for 2D stochastic constrained Navier-Stokes equations driven by Lévy noise in the Marcus canonical form}, J. Differential Equations, {\bf 302}, 64–-138.


\bibitem{Mao_2011} Mao, X. (2011). Stochastic Differential Equations, and Application, 2nd Edition, Woodhead Publishing.


\bibitem{Merton_1976} Merton, R. (1976). {\em Option pricing when underlying stock returns are discontinuous}, J. Financial Economics, {\bf 3}, pp. 125--144.

\bibitem{B3} Mohammadi, N., \& Panaretos, V. M. (2024). {\em Functional data analysis with rough sample paths?}. Journal of Nonparametric Statistics, {\bf 36}(1), 4-22.

\bibitem{Mohammadi2024} Mohammadi, N.,  
Leonardo, V. S. and Panaretos, V. M. (2024) {\em Nonparametric estimation for SDE with sparsely sampled paths: An FDA perspective}. Stochastic Processes and their Applications, {\bf 167}, 104239.


\bibitem{Oksendal_2003} \O{ksendal}, B. (2003). Stochastic Differential Equations, {\em An Introduction with Applications} 6th Edition, Springer-Verlag. 

\bibitem{Ramsey2005} Ramsey, J. O. and Silverman, B. W. (2005) {\em Functional Data Analysis},  Springer Series in Statistics.


\bibitem{Rudin} Rudin, W. (1987). {\em Real and Complex Analysis}, 3rd Edition. McGraw-Hill International Edition.

\bibitem{B6} Yao, F., M\"{u}ller, H. G., and Wang, J. L. (2005). {\em Functional data analysis for sparse longitudinal data}. Journal of the American Statistical Association, {\bf 100}(470), 577--590.


\end{thebibliography}

\begin{thebibliography}{99}



\bibitem{Bhar2024} Bhar, S. and Dhar, S. S. (2024). {\em Testing Independence of Infinite Dimensional Random Elements: A Sup-norm Approach}. \url{https://arxiv.org/abs/2301.00375}, Under revision. 



\bibitem{Dette2024} Dette, H. and Tang, J. (2024) {\em Statistical inference for function-on-function linear regression}, Bernoulli, {\bf 30}, 304--331. 





\bibitem{Zhang} Zhang, X., Wang, K.,  and Li, D. (2015). {\em Stochastic periodic solutions of stochastic differential equations driven by Lévy process}. Journal of Mathematical Analysis and Applications, {\bf 430}, 231--242.

\end{thebibliography}
\end{document}